\documentclass[12pt, onecolumn,a4paper]{quantumarticle}
\pdfoutput=1 
\usepackage[T1]{fontenc}

\usepackage{tikz}
\usetikzlibrary{positioning, calc}
\usetikzlibrary{calc}
\usetikzlibrary{arrows}
\usepackage{tikz-3dplot}
\usepackage{graphicx}
\usepackage{mdwlist}
\usepackage{color}
\usepackage{thmtools}
\usepackage{amssymb}
\usepackage{physics}
\usepackage{amsmath}
\usepackage{array}
\usepackage{doi}
\usepackage{url,hyperref}
\usepackage[capitalize,nameinlink]{cleveref}
\hypersetup{colorlinks={true},linkcolor={blue},citecolor=red}

\usepackage[numbers,sort&compress]{natbib}
\usetikzlibrary{fadings}
\usetikzlibrary{decorations.pathmorphing}
\usepackage{mathrsfs}
\usepackage{authblk}
\usepackage[mathscr]{euscript}
\usepackage{enumitem}
\usetikzlibrary{quantikz}

\DeclareMathAlphabet\mathbfcal{OMS}{cmsy}{b}{n}

\usepackage{cleveref}

\usetikzlibrary{decorations.pathreplacing}
\tikzset{snake it/.style={decorate, decoration=snake}}

\usetikzlibrary{decorations.pathreplacing,decorations.markings}

\tikzset{
    >=stealth',
    punkt/.style={
           rectangle,
           rounded corners,
           draw=black, very thick,
           text width=6.5em,
           minimum height=2em,
           text centered},
    pil/.style={
           ->,
           thick,
           shorten <=2pt,
           shorten >=2pt,},
  on each segment/.style={
    decorate,
    decoration={
      show path construction,
      moveto code={},
      lineto code={
        \path [#1]
        (\tikzinputsegmentfirst) -- (\tikzinputsegmentlast);
      },
      curveto code={
        \path [#1] (\tikzinputsegmentfirst)
        .. controls
        (\tikzinputsegmentsupporta) and (\tikzinputsegmentsupportb)
        ..
        (\tikzinputsegmentlast);
      },
      closepath code={
        \path [#1]
        (\tikzinputsegmentfirst) -- (\tikzinputsegmentlast);
      },
    },
  },
  mid arrow/.style={postaction={decorate,decoration={
        markings,
        mark=at position .5 with {\arrow[#1]{stealth'}}
      }}}
}

\usepackage{hyperref}
\usepackage{slashed}
\usepackage{float} 
\usepackage{graphicx} 
\usepackage{subcaption}

\usepackage{bbold}

\mathchardef\mhyphen="2D

\newcommand{\PClass}{\mathsf{P}}
\newcommand{\LClass}{\mathsf{L}}

\newcommand{\coNP}{\mathsf{coNP}}

\newcommand{\IP}{\mathsf{IP}}

\newcommand{\PP}{\mathsf{PP}}

\newcommand{\AM}{\mathsf{AM}}
\newcommand{\QAM}{\mathsf{QAM}}
\newcommand{\QIP}{\mathsf{QIP}}
\newcommand{\QPP}{\mathsf{QPP}}
\newcommand{\QNP}{\mathsf{QNP}}
\newcommand{\coQNP}{\mathsf{coQNP}}
\newcommand{\HVSZK}{\mathsf{HVSZK}}
\newcommand{\HVQSZK}{\mathsf{HVQSZK}}

\newcommand{\CDS}{\textnormal{CDS}}
\newcommand{\CDQS}{\textnormal{CDQS}}
\newcommand{\pp}{\textnormal{pp}}
\newcommand{\pc}{\textnormal{pc}}
\newcommand{\cc}{\textit{cc}}

\newtheorem{theorem}{Theorem}

\newtheorem{definition}[theorem]{Definition}

\newtheorem{lemma}[theorem]{Lemma}

\newtheorem{proposition}[theorem]{Proposition}
\newtheorem{remark}[theorem]{Remark}

\newenvironment{proof}[1][Proof]{\noindent\textbf{#1.}}{\ \rule{0.5em}{0.5em}}
\begin{document}

\title{Conditional disclosure of secrets with quantum resources}

\author[2]{Vahid R. Asadi}
\orcid{0000-0001-8354-7463}

\author[1,2]{Kohdai Kuroiwa}
\orcid{0000-0001-5555-1940}

\author[1,2]{Debbie Leung}
\orcid{}

\author[1,2]{Alex May}
\email{amay@perimeterinstitute.ca}
\orcid{0000-0002-4030-5410}

\author[1]{Sabrina Pasterski}
\orcid{0000-0003-3672-4169}

\author[1]{Chris Waddell}
\orcid{}

\affiliation[1]{Perimeter Institute for Theoretical Physics}
\affiliation[2]{Institute for Quantum Computing, Waterloo, Ontario}

\abstract{
The conditional disclosure of secrets (CDS) primitive is among the simplest cryptographic settings in which to study the relationship between communication, randomness, and security. 
CDS involves two parties, Alice and Bob, who do not communicate but who wish to reveal a secret $z$ to a referee if and only if a Boolean function $f$ has $f(x,y)=1$. 
Alice knows $x,z$, Bob knows $y$, and the referee knows $x,y$.
Recently, a quantum analogue of this primitive called CDQS was defined and related to $f$-routing, a task studied in the context of quantum position-verification. 
CDQS has the same inputs, outputs, and communication pattern as CDS but allows the use of shared entanglement and quantum messages.
We initiate the systematic study of CDQS, with the aim of better understanding the relationship between privacy and quantum resources in the information theoretic setting. 
We begin by looking for quantum analogues of results already established in the classical CDS literature. 
Concretely we establish the following.
\begin{itemize}
    \item \textbf{Closure:} Given a CDQS protocol for a function $f$, we construct CDQS protocols for the negation $\neg f$ of similar efficiency. Further, given CDQS protocols for functions $f_1$, $f_2$, we construct CDQS protocols to implement $f_1\wedge f_2$ and $f_1\vee f_2$ with cost not much larger than the sum of the costs for $f_1$ and $f_2$. 
    \item \textbf{Amplification:} Given a CDQS protocol with single qubit secrets and constant privacy and correctness errors, we construct CDQS schemes with $k$ qubit secrets and privacy and correctness errors of size $O(2^{-k})$, and whose communication and entanglement costs are increased by a factor of $k$.  
    \item \textbf{Lower bounds from $Q^*_{A\rightarrow B}(f)$:} We show that the quantum communication cost of a CDQS protocol for $f$ is lower bounded by the log of the one-way quantum communication cost with shared entanglement, $\CDQS(f)=\Omega( \log Q^*_{A\rightarrow B}(f))$.
    \item \textbf{Lower bounds from $\PP^{\cc}$:} Considering CDQS with perfect privacy, we lower bound the entanglement plus communication cost of CDQS linearly in terms of $\PP^{\cc}(f)$, the classical communication complexity of computing $f$ with unbounded error. 
    \item \textbf{Lower bounds from $\QIP[2]^{\cc}$:} Allowing constant privacy and correctness errors, we lower bound the communication cost of CDQS in terms of $\QIP[2]^{\cc}(f)$, the cost of a two message quantum interactive proof for the function $f$ in the communication complexity setting. 
    \item \textbf{Lower bounds from $\HVQSZK^{\cc}$:} Closely related to the above, we show that a similar lower bound on CDQS from the communication complexity of an honest verifier quantum statistical knowledge proof for $f$ lower bounds CDQS, up to logarithmic factors. 
\end{itemize}
Because of the close relationship to the $f$-routing position-verification scheme, our results have relevance to the security of these schemes. 
}
\maketitle

\pagebreak

\tableofcontents

\flushbottom

\section{Introduction}

In this article we study the conditional disclosure of secrets (CDS) primitive in a quantum setting. 
CDS involves three parties, Alice, Bob and a referee. 
Alice holds input $x\in \{0,1\}^n$, Bob holds $y\in\{0,1\}^n$, and the referee knows both $x$ and $y$. 
Additionally, Alice holds a secret string $z$, or (in a quantum context) a quantum system $Q$.  
Alice and Bob cannot communicate with one another, but can each send a message to the referee.
Given a Boolean function $f(x,y)$, the goal is for Alice and Bob's message to reveal the secret if and only if $f(x,y)=1$. 
In the classical setting, Alice and Bob's message consists of bits and they share a random string; in the quantum setting, they can send qubits and share entanglement. 
The CDS primitive is illustrated in \cref{fig:CDSandCDQS}. 

Classically, CDS has a number of applications in other aspects of cryptography: it was first studied and defined in the context of private information retrieval \cite{GERTNER2000592}, and has been applied in the context of attribute based encryption \cite{gay2015communication} and secret sharing \cite{applebaum2020power}. 
CDS also shares a number of connections to communication complexity \cite{applebaum2021placing}, and to another primitive known as private simultaneous message passing (PSM) \cite{feige1994minimal}. 
In particular, in both the classical and quantum settings lower bounds on CDS give lower bounds on PSM \cite{allerstorfer2024relating}. 
Perhaps most importantly, CDS is among the simplest settings in which we can study the relationship between privacy, communication, and randomness.

\begin{figure*}
    \centering
    \begin{subfigure}{0.45\textwidth}
    \centering
    \begin{tikzpicture}[scale=0.4]
    
    \draw[thick] (-5,-5) -- (-5,-3) -- (-3,-3) -- (-3,-5) -- (-5,-5);
    
    \draw[thick] (5,-5) -- (5,-3) -- (3,-3) -- (3,-5) -- (5,-5);
    
    \draw[thick] (5,5) -- (5,3) -- (3,3) -- (3,5) -- (5,5);
    
    \draw[thick, mid arrow] (4,-3) -- (4.5,3);
    
    \draw[thick, mid arrow] (-4,-3) to [out=90,in=-90] (3.5,3);
    
    \draw[thick,dashed] (-3.5,-5.5) -- (3.5,-5.5);
    \draw[black] plot [mark=*, mark size=3] coordinates{(-3.5,-5.5)};
    \draw[black] plot [mark=*, mark size=3] coordinates{(3.5,-5.5)};
    \node[below] at (0,-5.5) {$r$};
    
    \draw[thick] (-4.5,-6) -- (-4.5,-5);
    \node[below] at (-4.5,-6) {$x,z$};
    
    \draw[thick] (4.5,-6) -- (4.5,-5);
    \node[below] at (4.5,-6) {$y$};

    \node[left] at (0,1) {$m_0$};
    \node[right] at (4.5,0) {$m_1$};
    
    \draw[thick] (4,5) -- (4,6);
    \node[above] at (4,6) {$z$ iff $f(x,y)=1$};
    
    \end{tikzpicture}
    \caption{}
    \label{fig:CDQS}
    \end{subfigure}
    \hfill
    \begin{subfigure}{0.45\textwidth}
    \centering
    \begin{tikzpicture}[scale=0.4]
    
    \draw[thick] (-5,-5) -- (-5,-3) -- (-3,-3) -- (-3,-5) -- (-5,-5);
    
    \draw[thick] (5,-5) -- (5,-3) -- (3,-3) -- (3,-5) -- (5,-5);
    
    \draw[thick] (5,5) -- (5,3) -- (3,3) -- (3,5) -- (5,5);
    
    \draw[thick, mid arrow] (4,-3) -- (4.5,3);
    
    \draw[thick, mid arrow] (-4,-3) to [out=90,in=-90] (3.5,3);
    
    \draw[thick] (-3.5,-5) to [out=-90,in=-90] (3.5,-5);
    \draw[black] plot [mark=*, mark size=3] coordinates{(0,-7.05)};

    \node[left] at (0,1) {$M_0$};
    \node[right] at (4.5,0) {$M_1$};
    
    \draw[thick] (-4.5,-6) -- (-4.5,-5);
    \node[below] at (-4.5,-6) {$x, Q$};
    
    \draw[thick] (4.5,-6) -- (4.5,-5);
    \node[below] at (4.5,-6) {$y$};
    
    \draw[thick] (4,5) -- (4,6);
    \node[above] at (4,6) {Q iff $f(x,y)=1$};
    
    \end{tikzpicture}
    \caption{}
    \label{fig:PSQMintro}
    \end{subfigure}
    \caption{(a) A CDS protocol. Alice, on the lower left holds input $x\in \{0,1\}^n$ and a secret $z$ from alphabet $Z$. Bob, on the lower right, holds input $y\in \{0,1\}^n$. Alice and Bob can share a random string $r$. The referee, top right, holds $x$ and $y$. Alice sends a message $m_0(x,z,r)$ to the referee; Bob sends a message $m_1(y,r)$. The referee should learn $z$ iff $f(x,y)=1$ for some agreed on choice of Boolean function $f$. (b) A CDQS protocol. The communication pattern is as in CDS. The secret is now a quantum system $Q$, Alice and Bob can share a (possibly entangled) quantum state, and send quantum messages to the referee. The referee should be able to recover $Q$ iff $f(x,y)=1$.} 
    \label{fig:CDSandCDQS}
\end{figure*}

In this work we focus on this last aspect of CDS, but now in the quantum setting. 
We ask how privacy, quantum communication, and entanglement are related, specifically in the context of conditional disclosure of quantum secrets (CDQS), but with the larger goal in mind of understanding the relationship between these resources broadly in quantum cryptography. 
Further, we establish a number of relationships between CDQS and communication complexity, analogous to the classical results of \cite{applebaum2021placing}, which may be of interest to the theory of communication complexity. 
In the quantum setting, CDQS is closely related to a primitive known as $f$-routing \cite{kent2011quantum}.
$f$-routing is studied in the context of quantum position-verification \cite{chandran2009position,kent2011quantum,buhrman2014position}.
Towards a better understanding of CDQS, our focus in this work is on reproducing results on classical CDS in the quantum setting, or understanding when to not expect quantum analogues of classical results. 

Before proceeding, we define classical and quantum CDS more carefully, beginning with the classical case. 

\begin{definition}\label{def:CDS}
    A \textbf{conditional disclosure of secrets (CDS)} task is defined by a choice of function $f:\{0,1\}^{2n}\rightarrow \{0,1\}$.
    The scheme involves inputs $x\in \{0,1\}^{n}$ given to Alice, and input $y\in \{0,1\}^{n}$ given to Bob.
    Alice and Bob share a random string $r\in R$.
    Additionally, Alice holds a string $z$ drawn from distribution $Z$, which we call the secret. 
    Alice sends message $m_0(x,z,r)$ from alphabet $M_0$ to the referee, and Bob sends message $m_1(y,r)$ from alphabet $M_1$.  
    We require the following two conditions on a CDS protocol. 
    \begin{itemize}
        \item $\epsilon$\textbf{-correct:} There exists a decoding function $D(m_0,x,m_1,y)$ such that 
        \begin{align}
            \forall z\in Z,\,\forall \,(x,y) \in \{0,1\}^{2n} \,\,s.t.\,\,f(x,y)=1,\,\,\, \underset{r\leftarrow R}{\mathrm{Pr}}[D(m_0,x,m_1,y)=z] \geq 1-\epsilon\,. 
        \end{align}
        \item $\delta$\textbf{-secure:} There exists a simulator taking $(x,y)$ as input and producing a distribution $Sim$ on the random variable $M=M_0M_1$ such that
        \begin{align}
            \forall z\in Z,\,\forall \,(x,y) \in \{0,1\}^{2n} \,\,s.t.\,\, f(x,y)=0,\,\,\, \Vert Sim_{M}(x,y) - P_{M}(x,y,z) \Vert_1\leq \delta\,.
        \end{align}
        where $P_M(x,y,z)$ is the distribution on $M$ produced by the protocol on inputs $(x,y)$ and secret $z$. 
    \end{itemize}
\end{definition}
Considering the cost of a CDS protocol for a function $f$, we denote by $\CDS(f)$ the maximum of Alice and Bob's communication size, minimized over protocols $\Pi_{\epsilon,\delta}$ that complete the CDS with a fixed choice of $\epsilon$ and $\delta$, which we typically take to be $\epsilon=\delta=0.09$,\footnote{The choice of constant $\epsilon = \delta = 0.09$ in the definition of $\text{CDS}(f)$ appearing in (\ref{eq:def_cds_f}), and its quantum analogues in (\ref{eq:def_cdqs_f}), (\ref{eq:def_cdqs_bar_f}), is somewhat unconventional; for example, the definition in \cite{applebaum2021placing} assumes $\epsilon = \delta = 0.1$. 
It is necessary to take $\epsilon < 0.1$ in order for our proof of the amplification property of CDQS to hold. } 
\begin{align} \label{eq:def_cds_f}
    \CDS(f) = \min_{\Pi_{0.09,0.09}} (t_A + t_B)\,,
\end{align}
where $t_A$ and $t_B$ are the number of bits in Alice and Bob's messages respectively. 
We always take $t_A$ and $t_B$ to be maximized over inputs. 

To adapt this definition to the quantum setting, we need to be careful around the requirement that the classical security and correctness conditions hold for all choices of secret. 
In the quantum setting the secret string $z$ is now a quantum system $Q$, and we should have correctness and security for all input states. 
This is succinctly captured by phrasing the definition in terms of the diamond norm, which is a norm on the distance between quantum channels in the worst case over inputs. 

\begin{definition}\label{def:CDQS}
    A \textbf{conditional disclosure of quantum secrets (CDQS)} task is defined by a choice of function $f:\{0,1\}^{2n}\rightarrow \{0,1\}$, and a $d_Q$ dimensional Hilbert space $\mathbfcal{H}_Q$ which holds the secret.
    Alice and Bob share a resource system $\Psi_{LR}$, with $L$ held by Alice and $R$ held by Bob. 
    The task involves inputs $x\in \{0,1\}^{n}$ and system $Q$ given to Alice, and input $y\in \{0,1\}^{n}$ given to Bob.
    Alice sends message system $M_0$ to the referee, and Bob sends message system $M_1$. 
    Label the combined message systems as $M=M_0M_1$.
    Label the quantum channel defined by Alice and Bob's combined actions $\mathbfcal{N}_{Q\rightarrow M}^{x,y}$. 
    We put the following two conditions on a CDQS protocol. 
    \begin{itemize}
        \item $\epsilon$\textbf{-correct:} One 1 inputs, there exists a channel $\mathbfcal{D}^{x,y}_{M\rightarrow Q}$, called the decoder, such that the decoder approximately inverts the combined actions of Alice and Bob. That is 
        \begin{align}
            \forall (x,y)\in \{0,1\}^{2n} \,\,\, s.t. \,\, f(x,y)=1,\,\,\, \Vert\mathbfcal{D}^{x,y}_{M\rightarrow Q}\circ \mathbfcal{N}^{x,y}_{Q\rightarrow M} - \mathbfcal{I}_{Q\rightarrow Q}\Vert_\diamond \leq \epsilon\,.
        \end{align}
        \item $\delta$\textbf{-secure:} There exists a quantum channel $\mathbfcal{S}_{\varnothing \rightarrow M}^{x,y}$, called the simulator, which produces an output close to the one seen by the referee but which doesn't depend on the input state of $Q$. That is
        \begin{align}
            \forall (x,y)\in \{0,1\}^{2n} \,\,\, s.t. \,\, f(x,y)=0,\,\,\, \Vert\mathbfcal{S}_{\varnothing \rightarrow M}^{x,y} \circ \tr_Q - \mathbfcal{N}_{Q\rightarrow M}^{x,y}\Vert_\diamond \leq \delta\,.
        \end{align}
    \end{itemize}
\end{definition}
An alternative definition of CDQS would keep the secret classical, but still allow quantum resources. 
We can refer to this primitive as CDQS'. 
In fact, as noted in \cite{allerstorfer2024relating}, CDQS' and CDQS are equivalent, in the sense that for a given function $f$ they use nearly the same resources. 
To see why, notice that we obtain a CDQS' protocol from a CDQS protocol by fixing a basis for our input system $Q$. 
Conversely, we can obtain a CDQS protocol from a CDQS' protocol by a use of the quantum one-time pad: Alice applies a random Pauli $P_k$ to $Q$ and then sends the result to the referee. 
The choice of key $k$ is hidden in the CDQS' protocol. 
This allows the referee to recover $Q$ iff they can recover $k$, which ensures security and correctness of the CDQS.\footnote{Because CDQS' has the exact same inputs and outputs as the classical primitive it is a closer quantum analogue of CDS. 
We choose to start with the definition using a quantum secret because this will simplify a number of our proofs.}

We consider two measures of the cost of the protocol. 
First, we define the communication cost as the maximum over the number of qubits in Alice and Bob's messages, minimized over protocols that complete the task with security and correctness parameter $\epsilon=\delta=0.09$, 
\begin{align} \label{eq:def_cdqs_f}
    \CDQS(f) = \min_{\{\Pi_{0.09,0.09}\}} (n_{M_0} + n_{M_1})\,.
\end{align}
Again the message size of a protocol is defined to be the maximum of message sizes over choices of inputs. 
Second, we define a measure of the CDQS cost which also counts the size of the shared resource system, $\Psi_{LR}$, 
\begin{align} \label{eq:def_cdqs_bar_f}
    \overline{\CDQS}(f) = \min_{\{\Pi_{0.09,0.09}\}} (n_L+n_{M_0} + n_R+n_{M_1})\,.
\end{align}
Note that we allow Alice and Bob to apply arbitrary quantum channels to their systems, so the communication size and resource system size are not obviously related.

We can also study variants of CDS and CDQS where we require either perfect correctness ($\epsilon=0$), perfect security ($\delta=0$), or both. 
We add ``$\pc$'' to the cost function to denote the perfectly correct case, so that for example $\pc\CDS(f)$ denotes the communication cost of CDS for the function $f$ when requiring $\epsilon=0$, $\delta=0.09$. 
We similarly add ``$\pp$'' to label the perfectly private case, and just ``$\textnormal{p}$'' when we have both perfect correctness and perfect privacy. 
These can be combined with overlines to denote the shared resource plus communication cost, and with a Q to label the quantum case. 
Thus for example $\pp\CDQS(f)$ is the quantum communication cost in the perfectly private setting, and $\textnormal{p}\overline{\CDQS}(f)$ is the communication plus shared resource cost in the quantum setting with perfect privacy and perfect correctness. 

\vspace{0.2cm}
\noindent \textbf{Previous work:} CDQS was defined in \cite{allerstorfer2024relating}. 
There, one focus was on the relationship between CDQS and $f$-routing. 
To state this relationship, we first define the $f$-routing primitive. 

\begin{definition}\label{def:frouting}
    A \textbf{$f$-routing} task is defined by a choice of Boolean function $f:\{ 0,1\}^{2n}\rightarrow \{0,1\}$, and a $d$ dimensional Hilbert space $\mathcal{H}_Q$.
    Inputs $x\in \{0,1\}^{n}$ and system $Q$ are given to Alice, and input $y\in \{0,1\}^{n}$ is given to Bob.
    Alice and Bob exchange one round of communication, with the combined systems received or kept by Bob labelled $M$ and the systems received or kept by Alice labelled $M'$.
    Label the combined actions of Alice and Bob in the first round as $\mathbfcal{N}^{x,y}_{Q\rightarrow MM'}$. 
    The $f$-routing task is completed $\epsilon$-correctly if there exists a channel $\mathbfcal{D}^{x,y}_{M\rightarrow Q}$ such that,
    \begin{align}
        \forall (x,y)\in \{0,1\}^{2n} \,\,\, s.t. \,\, f(x,y)=1,\,\,\, \Vert\mathbfcal{D}^{x,y}_{M\rightarrow Q} \circ\tr_{M'} \circ\mathbfcal{N}^{x,y}_{Q\rightarrow MM'} -\mathbfcal{I}_{Q\rightarrow Q}\Vert_\diamond \leq \epsilon\,,
    \end{align}
    and there exists a channel $\mathbfcal{D}^{x,y}_{M'\rightarrow Q}$ such that
    \begin{align}
        \forall (x,y)\in \{0,1\}^{2n} \,\,\, s.t. \,\, f(x,y)=0,\,\,\,\Vert\mathbfcal{D}^{x,y}_{M'\rightarrow Q} \circ\tr_{M} \circ\mathbfcal{N}^{x,y}_{Q\rightarrow MM'} -\mathbfcal{I}_{Q\rightarrow Q}\Vert_\diamond \leq \epsilon\,.
    \end{align}
    In words, Bob can recover $Q$ if $f(x,y)=1$ and Alice can recover $Q$ if $f(x,y)=0$. 
\end{definition}
The $f$-routing task was defined in \cite{kent2011quantum} in the context of quantum position-verification (QPV). 
In a QPV scheme, a verifier sends a pair of messages to a prover, who should respond by computing a function of the input messages and returning them to the verifier.
The scheme should be thought of as occurring in a spacetime context, and the goal is for the prover to convince the verifier that they are performing computations within a specified spacetime region. 
In that context, performing $f$-routing means cheating in a corresponding QPV scheme. 

The following relationship between CDQS and $f$-routing was proven in \cite{allerstorfer2024relating}.
\begin{theorem}\label{thm:CDQSandfRouting}
    A $\epsilon$-correct $f$-routing protocol that routes $n$ qubits implies the existence of a $\epsilon$-correct and $\delta=2\sqrt{\epsilon}$-secure $\CDQS$ protocol that hides $n$ qubits using the same entangled resource state and the same message size. 
    A $\epsilon$-correct and $\delta$-secure $\CDQS$ protocol hiding secret $Q$ using a $n_E$ qubit resource state and $n_M$ qubit messages implies the existence of a $\max\{\epsilon,2\sqrt{\delta} \}$-correct $f$-routing protocol that routes system $Q$ using $n_E$ qubits of resource state and $4(n_M+n_E)$ qubits of message. 
\end{theorem}
From this relationship, upper and lower bounds on $f$-routing place corresponding upper and lower bounds on CDQS. 
We highlight however that the transformation between CDQS and $f$-routing preserves the size of the resource system, but not the resource system itself. 
In fact, CDQS for arbitrary $f$ can be performed using shared classical randomness, while $f$-routing even for some natural functions requires shared entanglement \cite{asadi2024rank}. 

CDQS is also related to its classical counterpart by the following statement. 
\begin{theorem}\label{thm:CDStoCDQS}
    An $\epsilon$-correct and $\delta$-secure $\CDS$ protocol hiding $2n$ bits and using $n_M$ bits of message and $n_E$ bits of randomness gives a $\CDQS$ protocol which hides $n$ qubits, is $2\sqrt{\epsilon}$-correct and $\delta$-secure using $n_M$ classical bits of message plus $n$ qubits of message, and $n_E$ classical bits of randomness.
\end{theorem}
From this theorem, we have that upper bounds on CDS place upper bounds on CDQS. 
Considering known upper bounds on CDS (and which therefore upper bound CDQS), we know that CDS can be performed for a function $f$ using randomness and communication upper bounded by the size of a secret sharing scheme with indicator function $f$, and by the size of a span program over $\mathbb{Z}_p$ computing $f$, with $p$ an arbitrary prime \cite{GERTNER2000592}. 
This last fact means that CDS and CDQS can be performed for all functions in the complexity class Mod$_p\LClass$ using polynomial resources. 
Considering the worst case, there is an upper bound of $2^{O(\sqrt{n\log n})}$ for all functions \cite{liu2017conditional}. 
Finally, there are specific functions believed to be outside of $\PClass$ but which have efficient secret sharing schemes and hence efficient CDS schemes \cite{beimel2001power}. 
Using results on $f$-routing, we also obtain an upper bound in terms of the size of a quantum secret sharing scheme with indicator function $f$ \cite{cree2022code}. 

Regarding lower bounds, CDQS inherits some lower bounds from results on $f$-routing. 
In \cite{bluhm2021position}, a linear lower bound on resource system size was proven for random choices of $f$, although their model assumes Alice and Bob apply unitary operations rather than fully general quantum channels. 
Recently a lower bound on entanglement in perfectly correct $f$-routing for some functions was proven in \cite{asadi2024rank}.
Concretely, their bound gives in the CDQS setting that
\begin{align}
    \pp\overline{\CDQS}(f) &= \Omega(\log \rank g|_f)\,, \nonumber \\
    \pc\overline{\CDQS}(f) &= \Omega(\log \rank g|_{\neg f})\,,
\end{align}
where $g|_f$ is a function which satisfies $g(x,y)=0$ iff $f(x,y)=0$. 
For some natural functions this leads to good lower bounds.
For instance when $f(x,y)$ is the equality function $g|_f$ has full rank, and when $f(x,y)$ is the not-equal function $g|_{\neg f}$ has full rank. 

We can also relate the above lower bounds on perfectly correct and perfectly private CDQS to the quantum non-deterministic communication complexity \cite{de2003nondeterministic}.
Suppose Alice holds input $x\in \{0,1\}^n$, Bob holds $y\in \{0,1\}^n$, and Alice and Bob communicate qubits to compute $f(x,y)$. 
The non-deterministic quantum communication complexity of $f$, $\QNP^{\cc}$, is defined to be the minimal number of qubits exchanged for Alice and Bob to both output $1$ with non-zero probability if and only if $f(x,y)=1$. 
In \cite{de2003nondeterministic}, the quantum non-deterministic complexity was characterized in terms of the non-deterministic rank of $f$. 
\begin{align}
    \pp\overline{\CDQS}(f) &= \Omega(\QNP^{\cc}(f))\,, \nonumber \\
    \pc\overline{\CDQS}(f) &= \Omega(\coQNP^{\cc}(f))\,.
\end{align}
The second bound is a quantum analogue of a lower bound on perfectly private CDS in terms of ${\coNP}^{\cc}(f)$ that was proven in \cite{applebaum2021placing}.

A closely related primitive to CDQS is private simultaneous message passing \cite{feige1994minimal}. 
A quantum analogue of PSM was introduced and studied in \cite{kawachi2021communication}. 
Classically, it is known that a PSM protocol for $f$ implies a CDS protocol for $f$ using similar resources. 
In this sense, CDS is a weaker notion than PSM. 
In \cite{allerstorfer2024relating}, the analogous statement for private simultaneous quantum message passing (PSQM) and CDQS was proven, so that again CDQS can be interpreted as a weaker primitive than PSQM. 
One implication is that our lower bounds on CDQS apply also to PSQM. 
Another comment is that PSQM and CDQS become trivial (can be performed with $O(1)$ sized messages) if Alice and Bob can share PR boxes \cite{allerstorfer2023making}.

\vspace{0.2cm}
\noindent \textbf{Our results:} We focus on three aspects of CDQS, which are closure, amplification, and lower bounds from communication complexity. 

We say CDS is closed under negation if for any Boolean function $f$ there is a CDS using similar resources for the negation of $f$.
Classically, closure of CDS was proven in \cite{applebaum2017conditional}. 
We show CDQS is also closed under negation, and in fact point out that the transformation from a protocol for $f$ to a protocol for $\neg f$ is both simpler and more efficient than the analogous classical result. 
In fact, the transformation is essentially trivial, following a standard purification argument. 
We give the formal statement and proof in \cref{sec:closure}. 
Further, we show that the CDQS costs (in both communication and entanglement) for $f_1\wedge f_2$ and $f_1 \vee f_2$ is not much larger than the sum of costs for $f_1$ and $f_2$. 
The argument uses a similar construction to the one used in the classical case \cite{applebaum2017conditional}. 

Classically \cite{applebaum2021placing} proved an amplification property of CDS: given a CDS protocol for $f$ with $\ell$ bit secrets, we can find a protocol for $f$ with $k \ell$ bit secrets for integer $k>0$ which uses a factor of $k$ more communication and randomness, and has correctness and security errors $O(2^{-k})$. 
In \cref{sec:amplification} we prove the analogous property for CDQS. 
As with closure, we find the quantum proof is simpler than the classical one, again due to a purification argument.
The purified view allows us to prove security of the amplified protocol by proving correctness of a purifying system, which is easier than considering security directly. 

In \cite{gay2015communication}, a relationship between CDS and one-way classical communication complexity was proven. 
In particular they showed
\begin{align}\label{eq:CDSandonewaycom}
    \CDS(f) \geq \frac{1}{4}\log (R_{A\rightarrow B}(f) + R_{B\rightarrow A}(f))\,,
\end{align}
where $R_{A\rightarrow B}(f)$ is the one-way classical communication complexity from Alice to Bob, and $R_{B\rightarrow A}(f)$ the one-way communication complexity from Bob to Alice. 
In \cref{sec:onewayCC} we prove an analogous lower bound in the quantum setting, 
\begin{align}
    \CDQS(f) = \Omega(\log Q^*_{A\rightarrow B}(f))\,.
\end{align}
The proof is simple and apparently unrelated to the classical proof. 
Further, the bound \eqref{eq:CDSandonewaycom} was understood to be tight: there exists an exponential gap between CDS cost and one-way communication complexity. 
We show that the same is true for quantum one-way communication complexity and the CDQS cost. 

In \cite{applebaum2021placing}, a number of further relationships between classical CDS and communication complexity were established. 
We reproduce three of these (with some modifications) in the quantum setting. 
First, \cite{applebaum2021placing} related perfectly private CDS and $\PP^{\cc}$ complexity,
\begin{align}\label{eq:CDSPPlowerbound}
    \pp\CDS(f) = \Omega(\PP^{\cc}(f))-O(\log(n))\,.
\end{align}
$\PP^{\cc}(f)$ measures the cost of outputting a variable $z$ which is biased towards the value of $f(x,y)$. 
The cost is defined by the total communication plus a term that grows as the bias becomes small. 
We produce a similar lower bound, 
\begin{align}
    \pp\overline{\CDQS}(f) = \Omega(\PP^{\cc}(f))\,.
\end{align}
Because explicit lower bounds are known for some functions against $\PP^{\cc}$, this also gives new explicit lower bounds for CDQS.
In particular this gives a linear lower bound on $\pp\overline{\CDQS}$ for the inner product function. 

Regarding fully robust CDS, \cite{applebaum2021placing} proved that
\begin{align}
    \CDS(f) \geq \IP^{\cc}[2](f)\,,
\end{align}
where $\IP^{\cc}[2](f)$ is the communication cost of a two-message interactive proof for the function $f$. 
We prove that\footnote{Note that in \cite{applebaum2021placing} $\IP^{\cc}[k]$ denotes a $k$ \emph{round} protocol, while here we let $\IP^{\cc}[k]$ denote a $k$ \emph{message} protocol. A round consists of a message from the verifiers to the prover and from the prover to the verifiers, so that each round consists of two messages. The notation we use is consistent with the convention in the quantum complexity literature.}
\begin{align}
    \CDQS(f) \geq \QIP^{\cc}[2](f)\,,
\end{align}
where the right hand side now is the cost of a quantum interactive proof in the communication complexity setting. 
Unfortunately, explicit lower bounds are not known against $\IP^{\cc}[2]$ or $\QIP^{\cc}[2]$, so this does not translate immediately to new explicit bounds. 
However, as with the classical case, the above bound does point to CDQS as a potentially easier setting to begin with in the context of trying to prove bounds against $\QIP^{\cc}[2]$. 

Finally, we strengthen our lower bound from $\QIP^{\cc}[2]$ by pointing out that it has a zero-knowledge property. 
We consider honest verifier statistically zero-knowledge proofs with two rounds.
See \cref{sec:QIPccbounds} for details. 
Classically, \cite{applebaum2021placing} proved the bound
\begin{align}
    \CDS(f) = \Omega\left(\frac{\overline{\HVSZK}^{\cc}[2](f)}{\log n} \right)\,,
\end{align}
where an overline indicates the communication plus randomness cost of the $\HVSZK$ protocol is being counted. 
We give a similar bound here, 
\begin{align}
    \CDQS(f) = \Omega\left(\frac{\HVQSZK^{\cc}[2](f)}{\log n} \right)\,.
\end{align}
Notice the overline is removed: in the quantum case we are only able to obtain a lower bound in terms of the communication cost alone. 

\cref{table:CDQScombounds} summarizes the known upper and lower bounds on communication cost in CDQS. Lower bounds follow directly from known lower bounds on the quantum one-way communication complexity of explicit functions such as Inner-Product \cite{cleve1998quantum, nayak2002communication}, Greater-Than \cite{anshu2017separation}, and Set-Disjointness \cite{braverman2019disjointness}.

\cref{table:CDQSentanglementbounds} gives the same when considering communication cost plus entanglement cost.

\begin{figure}
\begin{center}
\begin{tabular}{ |c c c c c|}
\hline
Function & $\textnormal{p}\CDQS$ & $\pc\CDQS$ & $\pp\CDQS$ & $\CDQS$ \\ 
 \hline
 \hline
 Equality & ${\Theta(1)}$ & $\Theta(1)$ & ${\Theta(1)}$ & $\Theta(1)$ \\  
 Non-Equality & $O(n)$ & $O(n)$ & $\Theta(1)$ & $\Theta(1)$ \\  
 Inner-Product &  $O(n)$, $\textcolor{blue}{\Omega(\log n)}$ & $O(n)$, $\textcolor{blue}{\Omega(\log n)}$ & $O(n)$, $\textcolor{blue}{\Omega(\log n)}$ & $O(n)$, $\textcolor{blue}{\Omega(\log n)}$ \\  
 Greater-Than & $O(n)$, $\textcolor{blue}{\Omega(\log n )}$ & $O(n)$, $\textcolor{blue}{\Omega(\log n)}$ & $O(n)$, $\textcolor{blue}{\Omega(\log n)}$ & $O(n)$, $\textcolor{blue}{\Omega(\log n)}$ \\  
 Set-Intersection & $O(n)$, $\textcolor{blue}{\Omega(\log n)}$ & $O(n)$, $\textcolor{blue}{\Omega(\log n)}$ & $O(n), \textcolor{blue}{\Omega(\log n)}$ & $O(n)$, $\textcolor{blue}{\Omega(\log n)}$ \\  
 Set-Disjointness & $O(n)$, $\textcolor{blue}{\Omega(\log n)}$ & $O(n), \textcolor{blue}{\Omega(\log n)}$ & $O(n)$, $\textcolor{blue}{\Omega(\log n)}$ & $O(n)$, $\textcolor{blue}{\Omega(\log n)}$ \\  
 \hline
\end{tabular}
\end{center}
\caption{Known upper and lower bounds on communication cost in CDQS. Blue entries are new to this work.
}\label{table:CDQScombounds}
\end{figure}

\begin{figure}
\begin{center}
\begin{tabular}{ |c c c c c|}
\hline
Function & \small{$p\overline{\CDQS}$} & \small{$\pc\overline{\CDQS}$} & \small{$\pp\overline{\CDQS}$} & \small{$\overline{\CDQS}$} \\ 
  \hline
  \hline
  Equality & ${\Theta(n)}$ & $\Theta(1)$ & ${\Theta(n)}$ & $\Theta(1)$ \\  
  Non-Equality & $\Theta(n)$ & $\Theta(n)$ & $\Theta(1)$ & $\Theta(1)$ \\  
  Inner-Product &  $\textcolor{blue}{\Theta(n)}$ & $O(n)$, $\textcolor{blue}{\Omega(\log n)}$ & $\textcolor{blue}{\Theta(n)}$ & $O(n)$, $\textcolor{blue}{\Omega(\log n)}$ \\
  Greater-Than & $\Theta(n)$ & $\Theta(n)$ & $\Theta(n)$ & $O(n)$, $\textcolor{blue}{\Omega(\log n)}$ \\  
  Set-Intersection & $\Theta(n)$ & $\Theta(n)$ & $\textcolor{blue}{\Omega(\log n)}$ & $O(n)$, $\textcolor{blue}{\Omega(\log n)}$ \\  
  Set-Disjointness & $\Theta(n)$ & $\textcolor{blue}{\Omega(\log n)}$ & $\Theta(n)$ & $O(n)$, $\textcolor{blue}{\Omega(\log n)}$ \\  
  \hline
 \end{tabular}
 \end{center}
 \caption{Known upper and lower bounds on entanglement plus communication cost in CDQS. Blue entries are new to this work.}\label{table:CDQSentanglementbounds}
 \end{figure}

\section{Some quantum information tools}

Quantum channels $\mathbfcal{N}_{A\rightarrow B}$ and $\mathbfcal{N}_{A\rightarrow E}^c$ are \emph{complimentary} if there is an isometry $\mathbf{V}_{A\rightarrow BE}$ such that
\begin{align}
    \mathbfcal{N}_{A\rightarrow B} &= \tr_{E}\circ \mathbf{V}_{A\rightarrow BE}( \cdot )\mathbf{V}_{A\rightarrow BE}\,, \nonumber \\
    \mathbfcal{N}^c_{A\rightarrow E} &= \tr_{B}\circ \mathbf{V}_{A\rightarrow BE}( \cdot )\mathbf{V}_{A\rightarrow BE}\,.
\end{align}
We will make use of the following property of complimentary channels, see e.g. \cite{wilde2013quantum}. 
\begin{remark}\label{remark:dilationsize}
    Given a channel $\mathbfcal{N}_{A\rightarrow B}$, there exists a complimentary channel $\mathbfcal{N}_{A\rightarrow E}^c$ with $d_E \leq d_A d_B$.
\end{remark}

A useful measure of how different two quantum channels are is provided by the diamond norm. 
\begin{definition}
The \textbf{diamond norm distance} between two channels $\mathbfcal{M},\mathbfcal{N}$ is defined by
\begin{align}
    \Vert\mathbfcal{M} - \mathbfcal{N}\Vert_\diamond = \max_{\ket{\Psi}_{RA}} \Vert I_R\otimes \mathbfcal{M}(\ket{\Psi}_{RA}) - I_R\otimes \mathbfcal{N}(\ket{\Psi}_{RA}) \Vert_1 \,
\end{align}
where $\dim R=\dim A$. 
\end{definition}
The diamond norm is a function of the probability of distinguishing two quantum channels in an operational setting \cite{wilde2013quantum}. 

The following theorem was proved in \cite{kretschmann2008information}.
\begin{theorem}\label{thm:newdecoupling}
    Let $\mathbfcal{N}_{A\rightarrow B}:\mathbfcal{L}(\mathbfcal{H}_A)\rightarrow \mathbfcal{L}(\mathbfcal{H}_B)$ be a quantum channel, and let $\mathbfcal{N}^c_{A\rightarrow E}$ be the complimentary channel. 
    Let $\mathbfcal{S}_{A \rightarrow E}$ be a completely depolarizing channel, which traces out the input and replaces it with a fixed state $\sigma_E$. 
    Then we have that
    \begin{align}
        \frac{1}{4} \inf_{\mathbfcal{D}_{B\rightarrow A}} \Vert \mathbfcal{D}_{B\rightarrow A} \circ \mathbfcal{N}_{A\rightarrow B}-\mathbfcal{I}_{A\rightarrow A}\Vert_{\diamond}^2 &\leq \Vert\mathbfcal{N}^c_{A\rightarrow E} - \mathbfcal{S}_{A \rightarrow E} \Vert_\diamond \nonumber \\
        &\leq 2 \inf_{\mathbfcal{D}_{B\rightarrow A}} \Vert\mathbfcal{D}_{B\rightarrow A} \circ \mathbfcal{N}_{A\rightarrow B} - \mathbfcal{I}_{A\rightarrow A} \Vert_\diamond^{1/2} \,, 
    \end{align}
    where the infimum is over all quantum channels $\mathbfcal{D}_{B\rightarrow A}$. 
\end{theorem} 
In words, this theorem says that if a channel is (nearly) invertible, it's complementary channel is (nearly) completely depolarizing, and vice versa. 
We will use this repeatedly in this work to relate the notion of security in CDQS to a notion of correctness for a purifying system.

\section{Basic properties}

\subsection{Closure under negation, AND, and OR}\label{sec:closure}

Closure of CDS under negation was shown in \cite{applebaum2017conditional}. 
We recall the exact statement here for convenience. 
\begin{theorem}\textbf{(From \cite{applebaum2017conditional})} Suppose that $f$ has a $\CDS$ with randomness complexity $\rho$ and communication complexity $t$ and privacy/correctness errors of $2^{-k}$. Then $\neg{f}=1-f$ has a $\CDS$ scheme with similar privacy/correctness error, and randomness/communication complexity $O(k^3\rho^2 t+ k^3\rho^3)$. 
\end{theorem}

\begin{figure*}
    \centering
    \begin{subfigure}{0.45\textwidth}
    \centering
        \begin{tikzpicture}[scale=0.4]
    
    \draw[thick] (-5,-5) -- (-5,-3) -- (-3,-3) -- (-3,-5) -- (-5,-5);
    \node at (-4,-4) {$\mathbfcal{N}^x$};
    
    \draw[thick] (5,-5) -- (5,-3) -- (3,-3) -- (3,-5) -- (5,-5);
    \node at (4,-4) {$\mathbfcal{N}^y$};
    
    \draw[thick] (5,5) -- (5,3) -- (3,3) -- (3,5) -- (5,5);
    
    \draw[thick] (4,-3) -- (4.5,3);
    
    \draw[thick] (-4,-3) to [out=90,in=-90] (3.5,3);
    
    \draw[thick] (-3.5,-5) to [out=-90,in=-90] (3.5,-5);
    \draw[black] plot [mark=*, mark size=3] coordinates{(0,-7.05)};
    \node[below right] at (-3.25,-5) {$L$};
    \node[below left] at (3.25,-5) {$R$};

    \node[left] at (0,1) {$M_0$};
    \node[right] at (4.5,0) {$M_1$};
    
    \draw[thick] (-4.5,-6) -- (-4.5,-5);
    \node[below] at (-4.5,-6) {$Q$};
    
    \draw[thick] (4,5) -- (4,6);
    \node[above] at (4,6) {Q iff $f(x,y)=1$};
    
    \end{tikzpicture}
    \caption{}
    \label{}
    \end{subfigure}
    \caption{A CDQS protocol, with all system labels and location of each quantum operation.} 
    \label{fig:CDQSprotocol}
\end{figure*}

The quantum version of this result is easier to prove and gives a more efficient transformation. 
\begin{theorem}\textbf{(Closure)} Suppose we have a $\epsilon$-correct and $\delta$-secure $\CDQS$ that reveals a $n_Q$ qubit system conditioned on a function $f$, uses $n_M$ message qubits, and a resource state $\ket{\Psi}_{LR}$ with $n_E=\log d_L=\log d_R$.
Then there exists a $\CDQS$ that reveals a $n_Q$ qubit system conditioned on $\neg f$, which uses at most $n_M+2n_E+n_Q$ message qubits, uses the same resource state, and is $\epsilon'=2\sqrt{\delta}$ correct and $\delta'=2\sqrt{\epsilon}$ secure. 
\end{theorem} 

\begin{proof}\,
    Consider the given CDQS protocol for $f$.
    Let Alice's channel be $\mathbfcal{N}^x_{QL\rightarrow M_0}$ and Bob's channel be $\mathbfcal{N}^y_{R\rightarrow M_1}$. 
    See figure \cref{fig:CDQSprotocol} for an illustration of the protocol. 
    In the new CDQS protocol for $\neg f$, have Alice apply the complimentary channel $(\mathbfcal{N}^x)^c_{QL\rightarrow M_0'}$ and Bob apply the complimentary channel $(\mathbfcal{N}^y)^c_{L\rightarrow M_0'}$. 
    This protocol uses the same resource system as the original, and using \cref{remark:dilationsize}, we have $n_{M'} = n_{M_0'} + n_{M_1'} \leq n_M + 2n_E + n_Q$ as needed.
    
    It remains to show correctness and security of this protocol. 
    To do this it is convenient to define 
    \begin{align}
        \mathbfcal{N}^{x,y}_{Q\rightarrow M}(\,\cdot_Q)\equiv\mathbfcal{N}^{x}_{QL\rightarrow M_0}\otimes \mathbfcal{N}^y_{R\rightarrow M_1}(\,\cdot_{Q} \otimes \Psi_{LR})\,.
    \end{align} 
    First consider security. 
    We consider $(x,y)$ which are zero instances of $\neg f(x,y)$, and hence one instances of $f(x,y)$. 
    Then by $\epsilon$ correctness of the CDQS for $f$, we have for all such $(x,y)$, there exists a decoder channel $\mathbfcal{D}^{x,y}_{M\rightarrow Q}$ such that
    \begin{align}
        \Vert\mathbfcal{D}^{x,y}_{M\rightarrow Q}\circ \mathbfcal{N}^{x,y}_{Q\rightarrow M} - \mathbfcal{I}_{Q\rightarrow Q}\Vert_\diamond \leq \epsilon \,.
    \end{align} 
    Then by \cref{thm:newdecoupling} we get that there exists a completely depolarizing channel $\mathbfcal{S}_{Q\rightarrow M'}^{x,y}$ such that
    \begin{align}
        \Vert(\mathbfcal{N}^{x,y})^c_{Q\rightarrow M'} - \mathbfcal{S}^{x,y}_{Q \rightarrow M'} \Vert_\diamond \leq 2\sqrt{\epsilon} \,.
    \end{align}
    But in the new CDQS protocol the protocol implements $(\mathbfcal{N}^{x,y})^c_{Q\rightarrow M'}$, so this is exactly $\delta'=2\sqrt{\epsilon}$ security. 

    Next we establish correctness. 
    Consider an input pair $(x,y)\in (\neg f)^{-1}(1)$, so $(x,y)\in f^{-1}(0)$. 
    Then by security of the original CDQS, we have that there exists a completely depolarizing channel $\mathbfcal{S}_{Q \rightarrow M}^{x,y}=\mathbfcal{S}_{\varnothing \rightarrow M}^{x,y} \circ \tr_Q$ such that
    \begin{align}
        \Vert\mathbfcal{S}_{\varnothing \rightarrow M}^{x,y} \circ \tr_Q - \mathbfcal{N}_{Q\rightarrow M}^{x,y}\Vert_\diamond \leq \delta \,.
    \end{align}
    But again by \cref{thm:newdecoupling} this mean there exists a decoding channel $\mathbfcal{D}^{x,y}_{M'\rightarrow Q}$ such that
    \begin{align}
        \Vert\mathbfcal{D}^{x,y}_{M'\rightarrow Q}\circ (\mathbfcal{N}^{x,y})^{c}_{Q\rightarrow M'} - \mathbfcal{I}_{Q\rightarrow Q} \Vert_\diamond \leq 2\sqrt{\delta} \,,
    \end{align}
    which is $\epsilon'=2\sqrt{\delta}$ correctness of the CDQS for $\neg f$. 
\end{proof}

We can also show that the cost of implementing a function which is the AND or OR of two functions $f_1$, $f_2$, cannot be much larger than the sum of the costs for $f_1$ and $f_2$. 
Thus in this sense CDQS is closed under taking AND or OR's. 

\begin{theorem}
    Let $f_1$ and $f_2$ be functions which can be implemented in CDQS protocols with entanglement cost $E_1$ and $E_2$, and communication cost $q_1$ and $q_2$. 
    Further, suppose both protocols are $\epsilon$-correct and $\delta$-secure. 
    Then there is a protocol for $f_1\wedge f_2$ which is $2\epsilon$ correct and $\delta$-secure with communication cost at most $q_1+q_2$ and entanglement cost is at most $E_1+E_2$. 
    Similarly, there is a CDQS protocol for $f_1 \vee f_2$ which is $\epsilon$ correct and $2\delta$ secure, and has communication cost at most $q_1+q_2+c$ for $c$ a constant and entanglement cost at most $E_1+E_2$. 
\end{theorem}

\begin{proof}\,
    Begin with the case of $f_1\wedge f_2$. Record the secret, $Q$, into a $((2,2))$ quantum secret sharing scheme. Call the resulting shares $Q_1$ and $Q_2$. Input $Q_1$ to the CDQS protocol for $f_1$, and $Q_2$ to the CDQS protocol for $f_2$. 

    Correctness is clear: if $f=f_1\wedge f_2=1$ then $f_1=1$ and $f_2=1$, so the referee recovers both $Q_1$ and $Q_2$, so they can recover $Q$ from the quantum secret sharing scheme. 
    To determine the precise correctness parameter, we use that by $\epsilon$ correctness of the protocols for $f_1$ and $f_2$, 
    \begin{align}
        \Vert\mathbfcal{D}_{M_1\rightarrow Q_1}^{x,y}\circ \mathbfcal{N}^{x,y}_{Q_1\rightarrow M_1} - \mathbfcal{I}_{Q_1\rightarrow Q_1}\Vert_\diamond &\leq \epsilon \nonumber \\
        \Vert\mathbfcal{D}_{M_2\rightarrow Q_2}^{x,y}\circ \mathbfcal{N}^{x,y}_{Q_2\rightarrow M_2} - \mathbfcal{I}_{Q_2\rightarrow Q_2}\Vert_\diamond &\leq \epsilon
    \end{align}
    where the $\mathbfcal{N}^{x,y}_{Q_i\rightarrow M_i}$ represent the combined actions of Alice and Bob in the CDQS protocol for $f_i$, and $\mathbfcal{D}_{M_1\rightarrow Q_1}^{x,y}$ the decoding channels for those protocols. 
    Further, call the encoding channel for the secret sharing scheme $\mathbfcal{E}_{Q\rightarrow Q_1Q_2}$ and decoding channel $\mathbfcal{R}_{Q_1Q_2\rightarrow Q}$, and note that $\mathbfcal{R}_{Q_1Q_2\rightarrow Q} \circ \mathbfcal{E}_{Q\rightarrow Q_1Q_2} = \mathbfcal{I}_{Q}$.
    Then to establish correctness of the combined protocol we need to bound the expression
    \begin{align}
        \Vert\mathbfcal{R}_{Q_1Q_2\rightarrow Q}\circ  (\mathbfcal{D}_{M_1\rightarrow Q_1}^{x,y}\otimes \mathbfcal{D}_{M_2\rightarrow Q_2}^{x,y})\circ (\mathbfcal{N}^{x,y}_{Q_1\rightarrow M_1}\otimes \mathbfcal{N}^{x,y}_{Q_2\rightarrow M_2})\circ \mathbfcal{E}_{Q\rightarrow Q_1Q_2} -\mathbfcal{I}_{Q\rightarrow Q}\Vert \nonumber 
    \end{align}   
    To do this, define $\mathbfcal{C}_{Q_1\rightarrow Q_1}^{x,y} =\mathbfcal{D}_{M_1\rightarrow Q_1}^{x,y}\circ \mathbfcal{N}^{x,y}_{Q_1\rightarrow M_1}$ and $\mathbfcal{C}_{Q_2\rightarrow Q_2}^{x,y} =\mathbfcal{D}_{M_2\rightarrow Q_2}^{x,y}\circ \mathbfcal{N}^{x,y}_{Q_2\rightarrow M_2}$, and re-express the above as
    \begin{align}    
        \Vert\mathbfcal{R}_{Q_1Q_2\rightarrow Q}\circ  (\mathbfcal{C}^{x,y}_{Q_1\rightarrow Q_1}\otimes \mathbfcal{C}^{x,y}_{Q_2\rightarrow Q_2})\circ \mathbfcal{E}_{Q\rightarrow Q_1Q_2} -\mathbfcal{I}_{Q\rightarrow Q}\Vert 
    \end{align}
    Define $X_i=\mathbfcal{C}^{x,y}_{Q_i\rightarrow Q_i}-\mathbfcal{I}_{Q_i\rightarrow Q_i}$, and drop the system label subscripts for convenience, to obtain
    \begin{align}    
        \Vert\mathbfcal{R}\circ  (X_1\otimes \mathbfcal{C}_2)\circ \mathbfcal{E}+\mathbfcal{R}\circ (\mathbfcal{I}\otimes X_2)\circ \mathbfcal{E} \Vert_\diamond  &\leq \Vert (X_1\otimes \mathbfcal{C}_2)+ (\mathbfcal{I}\otimes X_2)\Vert_\diamond \nonumber \\
        &\leq \Vert X_1\otimes \mathbfcal{C}_2\Vert_\diamond + \Vert\mathbfcal{I}\otimes X_2\Vert_\diamond \nonumber \\
        &= \Vert X_1\Vert_\diamond \Vert\mathbfcal{C}_2\Vert_\diamond + \Vert X_2\Vert_\diamond \nonumber \\
        &\leq 2\epsilon
    \end{align}
    so that the protocol is $2\epsilon$ correct. 

    To establish security, we use that $\delta$ security of the protocols for $f_1$ and $f_2$ imply the existence of simulator channels $\mathbfcal{S}_{\varnothing \rightarrow M_1}^{x,y}, \mathbfcal{S}_{\varnothing \rightarrow M_2}^{x,y}$ such that
    \begin{align}
            \Vert\mathbfcal{S}_{\varnothing \rightarrow M_1}^{x,y} \circ \tr_{Q_1} - \mathbfcal{N}_{Q_1\rightarrow M_1}^{x,y}\Vert_\diamond \leq \delta \nonumber \\
            \Vert\mathbfcal{S}_{\varnothing \rightarrow M_2}^{x,y} \circ \tr_{Q_2} - \mathbfcal{N}_{Q_2\rightarrow M_2}^{x,y}\Vert_\diamond \leq \delta
    \end{align}
    Then we claim $\mathbfcal{S}_{\varnothing \rightarrow M_1}^{x,y}\otimes  \mathbfcal{S}_{\varnothing \rightarrow M_2}^{x,y}$ is a good simulator channel for the protocol defined above for $f_1\wedge f_2$. 
    To establish this, we need to show that when $f=0$, so that at least one of $f_1$ or $f_2$ are zero, there is a good simulator channel for the system received by the referee. 
    Let's suppose without loss of generality that $f_1=0$. 
    Then we claim
    \begin{align}
        \Vert((\mathbfcal{S}_{\varnothing \rightarrow M_1}^{x,y}\circ \tr_{Q_1})\otimes  \mathbfcal{N}_{Q_2\rightarrow M_2}^{x,y}) \circ \mathbfcal{Y}_{\varnothing \rightarrow Q_2}\circ \tr_Q- (\mathbfcal{N}_{Q_1\rightarrow M_1}^{x,y}\otimes \mathbfcal{N}_{Q_2\rightarrow M_2}^{x,y})\circ \mathbfcal{E}_{Q\rightarrow Q_1Q_2} \Vert_\diamond \leq \delta \nonumber 
    \end{align}
    so that $((\mathbfcal{S}_{\varnothing \rightarrow M_1}^{x,y}\circ \tr_{Q_1})\otimes  \mathbfcal{N}_{Q_2\rightarrow M_2}^{x,y}) \circ \mathbfcal{Y}_{\varnothing \rightarrow Q_2}$ is a simulator channel. 
    We will use that security of the secret sharing scheme means there exists a channel $\mathbfcal{Y}_{\varnothing\rightarrow Q_1}$ such that $\tr_{Q_1} \circ\, \mathbfcal{E}_{Q\rightarrow Q_1Q_2}=\mathbfcal{Y}_{\varnothing \rightarrow Q_2}\circ \tr_Q$.
    Then, 
    \begin{align}
        \Vert&((\mathbfcal{S}_{\varnothing \rightarrow M_1}^{x,y}\circ \tr_{Q_1})\otimes  \mathbfcal{N}_{Q_2\rightarrow M_2}^{x,y}) \circ \mathbfcal{Y}_{\varnothing \rightarrow Q_2}\circ \tr_Q- (\mathbfcal{N}_{Q_1\rightarrow M_1}^{x,y}\otimes \mathbfcal{N}_{Q_2\rightarrow M_2}^{x,y})\circ \mathbfcal{E}_{Q\rightarrow Q_1Q_2} \Vert_\diamond \nonumber \\
        &\leq \Vert(\mathbfcal{S}_{\varnothing \rightarrow M_1}^{x,y}\circ \tr_{Q_1}) \circ \mathbfcal{Y}_{\varnothing \rightarrow Q_2}\circ \tr_Q - \mathbfcal{N}_{Q_1\rightarrow M_1}^{x,y}\circ \mathbfcal{E}_{Q\rightarrow Q_1Q_2} \Vert_\diamond \nonumber \\
        &\leq \Vert(\mathbfcal{S}_{\varnothing \rightarrow M_1}^{x,y}\circ \tr_{Q_1}) \circ \mathbfcal{E}_{Q\rightarrow Q_1Q_2} - \mathbfcal{N}_{Q_1\rightarrow M_1}^{x,y}\circ \mathbfcal{E}_{Q\rightarrow Q_1Q_2} \Vert_\diamond \nonumber \\\
        &\leq \Vert(\mathbfcal{S}_{\varnothing \rightarrow M_1}^{x,y}\circ \tr_{Q_1}) - \mathbfcal{N}_{Q_1\rightarrow M_1}^{x,y}\Vert_\diamond \nonumber \\
        &\leq \delta
    \end{align}
    as claimed. 

    To build a protocol for $f=f_1\vee f_2$, we record $Q$ into a secret sharing scheme with three shares such that any two suffice to recover the secret.
    Call the shares $Q_0Q_1Q_2$. 
    Then Alice always sends $Q_0$ (contributing the $+c$ to the total communication cost) and inputs $Q_1$ to a CDQS protocol for $f_1$, and $Q_2$ to a CDQS protocol for $f_2$. 
    The analysis to establish security and correctness follows the same strategy as above; we omit it for brevity. 
\end{proof}

\subsection{Amplification}\label{sec:amplification}

We will prove a quantum analogue of Theorem 12 from \cite{applebaum2017conditional}, which we state below. 

\begin{theorem}\label{thm:amplification}
    Let F be a $\CDS$ for a function $f$ that supports one bit secrets with correctness error $\epsilon=0.09$ and privacy error $\delta=0.09$. Then for every integer $k$ there exists a $\CDS$ G for $f$ with $k$-bit secrets and privacy and correctness errors of $2^{-\Omega(k)}$. The communication and randomness complexity of $G$ is larger than that of F by a factor of $k$. 
\end{theorem}
This theorem has applications in relating CDS and communication complexity, and we will need a quantum analogue of this result for the same purpose. 

To reproduce this for CDQS, we first of all need to recall the existence of ``good'' quantum error correcting codes. 
By a good code we mean one with distance $t$ and number of physical qubits $n_P$ both linear in the number of logical qubits.
Existence of these codes is established in \cite{calderbank1996good}.  
We summarize the parameters of their construction here. 
\begin{remark}\label{remark:goodcodes}
    There exist quantum codes with $k$ logical qubits, $m$ physical qubits, and correcting arbitrary errors on $t<m/2$ qubits with 
    \begin{align}
        \frac{k}{m} = 1- 2 H_2\left(\frac{2t}{m} \right) \,,
    \end{align}
    where $H_2$ is the binary entropy function. 
\end{remark}
Taking $t = \alpha m$ to be a constant fraction of $m$, we find there are quantum codes with $k/m=\beta$ a constant. 

An error-correcting code that corrects arbitrary errors on $t$ qubits will also do well in correcting small errors on all qubits. 
This is expressed in the next theorem, which we reproduce from \cite{gottesman2024surviving}.
\begin{theorem}\label{eq:iidnoise}
    Let $\mathcal{I}$ be the one-qubit identity channel and $\mathbfcal{E}=\otimes_{i=1}^m \mathbfcal{E}_i$ be an $m$-qubit independent error channel, with $\Vert\mathbfcal{E}_i-\mathbfcal{I} \Vert_\diamond < \epsilon< \frac{t+1}{m-t-1}$, and let $\mathbfcal{U}$ and $\mathcal{D}$ be the encoder and decoder for a QECC with $m$ physical qubits that corrects $t$-qubit errors. Then
    \begin{align}
        \Vert \mathbfcal{D}\circ \mathbfcal{E}\circ \mathbfcal{U} - \mathbfcal{I}\Vert_\diamond < 2 \binom{m}{t+1} (e \epsilon)^{t+1} \,.
    \end{align}
\end{theorem}
Combined with the existence of codes that correct $t=\alpha m$ qubit errors, this theorem allows for exponential suppression of errors. 
We make use of this fact in the next theorem, showing amplification for $f$-routing. 

\begin{theorem}\label{thm:fRamplification}
    Let $F_Q$ be an $f$-routing protocol for a function $f$ that supports one qubit input systems with correctness error $\epsilon=0.09$, communication cost $c$, and entanglement cost $E$. Then for every integer $k$ there exists an $f$-routing protocol $G_{Q'}$ for $f$ with $k$-qubit secrets, privacy and correctness errors of $2^{-\Omega(k)}$, communication cost $O(k c)$, and entanglement cost $O(k E)$. 
\end{theorem}

\begin{proof}\,
    We let $Q'$ be the $k$-qubit input to the $f$-routing protocol $G_{Q'}$. 
    Alice encodes $Q'$ into an error correcting code with $k$ logical qubits, $m=k/\beta$ physical qubits, and is able to tolerate $t=\alpha m$ errors with $\alpha = 0.495$, $\beta \approx 0.838$ constant. 
    Such codes exist by \cref{remark:goodcodes}. 
    Let the encoded qubits be systems $\{S_i\}_{i=1}^m$. 
    Alice and Bob run an instance of $F_Q$ on each share $S_i$. 
    At the output location specified by $f(x,y)$, the receiving party decodes the error-correcting code and attempts to recover $Q$.
    The combined operations of encoding, running $F_Q$ on each share of the code, then decoding define the new $f$-routing protocol $G_{Q'}$.
    
    We claim that $\epsilon=0.09$ correctness of each $F_Q$ implies $2^{-\Omega(k)}$ correctness of $G_{Q'}$. 
    The action of the $m$ instances of the $f$-routing protocol can be captured by the channel $\mathbfcal{E}=\otimes_{i=1}^m\mathbfcal{E}_i$, with correctness of each instance giving that 
    \begin{align}
        \Vert\mathbfcal{E}_i-\mathbfcal{I} \Vert_\diamond \leq \epsilon = 0.09 \,.
    \end{align}
    Now use \cref{eq:iidnoise} with parameters $\epsilon=0.09$, $t=\alpha m$. 
    Using that $\binom{m}{t+1} \leq 2^{m H_2 \left( (t+1) / m \right)}$ with $H_2(\cdot)$ the binary entropy function, we find that 
    \begin{align}
        \Vert\mathbfcal{D}\circ \mathbfcal{E}\circ \mathbfcal{U} - \mathbfcal{I} \Vert_\diamond < 
        C 2^{-\gamma m} \,,
    \end{align}
    with $C>0$ and $\gamma \approx 5.5 \times 10^{-3}$ when $\alpha = 0.495$, as we have assumed. 
    Since $m=\Theta(k)$, this gives the needed correctness. 

    Finally we note that the new $f$-routing protocol uses $m$ copies of the original protocol, with $m=O(k)$, so the communication and entanglement costs are increased by factors of $k$ as claimed. 
\end{proof}

Amplification for CDQS follows from the above along with \cref{thm:CDQSandfRouting} relating $f$-routing and $\CDQS$. 
We record this fact as the following theorem. 

\begin{theorem}
    Let $F_Q$ be a $\CDQS$ for a function $f$ that supports one qubit  secrets with correctness error $\delta=0.09$ and privacy error $\epsilon=0.09$, has communication cost $c$, and entanglement cost $E$. 
    Then for every integer $k$, there exists a $\CDQS$ $G_Q$ for $f$ with $k$-qubit secrets, privacy and correctness errors of $2^{-\Omega(k)}$, and communication and entanglement complexity of size $O(k c)$ and $O(k E)$, respectively. 
\end{theorem}

\begin{proof}\, 
Follows by using \cref{thm:CDQSandfRouting} to transform the CDQS into an $f$-routing protocol, applying the amplification result from \cref{thm:fRamplification}, then using \cref{thm:CDQSandfRouting} to turn the $f$-routing into a CDQS protocol again. 
\end{proof}

\section{Lower bounds}

\subsection{Lower bounds from one-way quantum communication complexity}\label{sec:onewayCC}

From \cite{gay2015communication}, we have the lower bound
\begin{align}
    \CDS(f) \geq \frac{1}{4} \log (R_{A\rightarrow B}(f) + R_{B\rightarrow A}(f)) \,,
\end{align}
so that the classical CDS communication cost is lower bounded by the log of the one-way communication complexity. 

We will prove a similar lower bound in the quantum setting. 
To do so, we rely on a reduction that involves Alice performing state tomography on the message system in a CDQS protocol. 
We make use of the following result on state tomography. 

\begin{theorem}[Reproduced from~\cite{ODonnell:2015lma}]\label{thm:tomography} Given $k=O(\log(1/\delta)d^2/\epsilon^2)$ copies of an unknown state $\rho$, there is a strategy that produces an estimator state $\hat{\rho}$ which is $\epsilon$ close to $\rho$ in trace distance with probability $1-\delta$.
\end{theorem}

Our reduction is to the quantum one-way communication complexity, with shared entanglement allowed.  
We define this next 
\begin{definition}[Quantum one-way communication complexity] Let $f:\{0,1\}^n\times \{0,1\}^n\rightarrow \{0,1\}$ and $\delta\in[0,1]$.
A one-way communication protocol $P_\delta$ for $f$ is defined as follows.
Alice receives $x \in\{0,1\}^n$ as input and produces quantum system $M_A$ as output, which she sends to Bob. 
Bob receives $y \in\{0,1\}^n$ and $M_A$, and outputs a bit $z$.
The protocol is $\delta$-correct if $\Pr[z=f(x,y)]\geq 1-\delta$. 

The quantum one-way communication complexity of $f$, $Q_{\delta,A\rightarrow B}(f)$ is defined as the minimum number of qubits in $M_A$ needed to achieve $\delta$-correctness. 
We write $Q_{A\rightarrow B}(f)\equiv Q_{\delta=0.09,A\rightarrow B}(f)$
Similarly, we can define $ Q_{\delta,A\rightarrow B}^{*}(f)$ where Alice and Bob are allowed pre-shared entanglement. 
\end{definition}

We are now ready to prove the main theorem of this section. 

\begin{theorem} \label{thm:oneway}
The one-way quantum communication complexity of $f$ and the communication cost of a $\CDQS$ protocol for $f$ are related by
\begin{align}
    \CDQS(f) = \Omega(\log Q_{B\rightarrow A}^{*}(f))\,.
\end{align} 
\end{theorem}
\begin{proof}\, Beginning with a CDQS protocol, we will build a one-way quantum communication protocol. 
In the CDQS, we let Bob's output system be called $M_1$ and Alice's output be $M_0$, and label $M_0M_1=M$. 
To define the one-way protocol, we have Alice and Bob share $k$ copies of the resource system for the CDQS, and repeat the operations they would implement in the CDQS protocol on each copy. 
Concretely, Alice inputs the $Q$ subsystem of a maximally entangled state $\Psi^+_{\bar{Q}Q}$ to the CDQS protocol, and stores the $\bar{Q}$ system.  
In each of the $k$ instances, Bob takes his output $M_1$ and sends it to Alice. 
Alice then performs tomography on her $k$ copies of the state which, we claim, allows her to determine the value of $f(x,y)$ with high probability. 
 
To see why, observe that if $f(x,y)=0$, then $\bar{Q}$ is decoupled from the message system $M$, at least approximately. 
In particular the security statement of the CDQS implies there is a channel $\mathbfcal{S}^{xy}_{\varnothing \rightarrow M}$ such that
\begin{align}
\delta \geq \Vert\mathbfcal{S}_{\varnothing \rightarrow M}^{xy}\circ \tr_Q(\Psi^+_{\bar{Q}Q}) - \mathbfcal{N}^{xy}_{Q\rightarrow M} (\Psi^+_{\bar{Q}Q})\Vert_1 = \left\Vert \frac{\mathcal{I}_{\bar{Q}}}{2}\otimes \sigma_{M} - \rho_{\bar{Q}M}\right\Vert_1 \,.
\end{align}
Meanwhile, if $f(x,y)=1$, we have that there exists a recovery channel $\mathbfcal{D}^{x,y}_{M\rightarrow Q}$ such that
\begin{align}
    \Vert\mathbfcal{D}^{x,y}_{M\rightarrow Q} (\rho_{\bar{Q}M}) -\Psi^+_{\bar{Q}Q}\Vert_1 \leq \epsilon \,,
\end{align}
so that, for any product state $\sigma_{\bar{Q}}\otimes \sigma_{M}$
\begin{align}
    \Vert\rho_{\bar{Q}M} - \sigma_{\bar{Q}}\otimes \sigma_{M} \Vert_1 &\geq \Vert\mathbfcal{D}^{x,y}_{M\rightarrow Q} (\rho_{\bar{Q}M}) - \sigma_{\bar{Q}}\otimes \mathbfcal{D}^{x,y}_{M\rightarrow Q}(\sigma_{M})\Vert_1 \nonumber \\
    &\geq \Vert\Psi^+_{\bar{Q}Q} - \sigma_{\bar{Q}}\otimes \sigma_{Q}'\Vert_1 - \epsilon \nonumber \\
    &\geq 2\left(1 - \frac{1}{\sqrt{d_Q}}\right) - \epsilon \,,
\end{align}
where the last line follows by upper bounding the fidelity of $\Psi^+_{\bar{Q}Q}$ with any product state and then applying Fuchs--van de Graaf inequality. 
Using $\epsilon=0.09$ and $d_Q=2$, the lower bound is $\approx 0.496$. 
Summarizing, we have that for $\epsilon=\delta=0.09$, the one-norm distance from the product state is less than $0.09$ if $f(x,y)=0$ and larger than $0.496$ if $f(x,y)=1$. 
Consequently, if Alice can determine $\rho$ to within constant error from her samples she can determine the value of $f(x,y)$. 

To do this, Alice applies the tomography protocol of \cref{thm:tomography} to her $k$ samples. 
Using $O(\log(1/\delta)d^2/\tilde{\epsilon}^2)$ samples, she can with probability $1-\delta$ determine a density matrix $\hat{\rho}$ which is guaranteed to be $\tilde{\epsilon}$ close to $\rho$. 
Taking $\tilde{\epsilon}$ small enough ensures $\hat{\rho}'$ is small enough to distinguish if $\rho$ is within $\epsilon=0.09$ in trace distance to product, or further than $\epsilon=0.496$ away from product, so that Alice can determine $f(x,y)$ with probability $1-\delta$.  

Since $\tilde{\epsilon}$ is a constant, this one-way quantum communication protocol uses $k=O(d^2)$ qubits of message, where $d$ is the dimension of Bob's message. 
In terms of the CDQS cost, this leads to
\begin{align}
     O(2^{2\CDQS(f)}) = Q^*_{ B\rightarrow A}(f) \,.
\end{align}
Equivalently, 
\begin{align}
    \CDQS(f) = \Omega (\log Q^*_{B\rightarrow A}(f)) \,.
\end{align}
\end{proof}

We can also reduce to a setting without shared entanglement, at the expense of now bounding the sum of the entanglement and communication used in the CDQS and allowing two-way communication in the communication scenario.
In particular, we can have Alice locally prepare the entangled state used in the CDQS and send Bob's share to him, then have Bob apply his first round CDQS operation and send the output back to Alice.
Using this reduction we obtain
\begin{align}
    \overline{\CDQS}(f) = \Omega(\log Q_\delta(f)) \,.
\end{align}
For the inner product function, \cite{cleve1998quantum} gives a linear lower bound on $Q^*_{B\rightarrow A}$. 
As another variant, we also note in \cref{sec:CliffordCDQS} that if we restrict to CDQS protocols that implement only Clifford operations, then we obtain a lower bound of $\Omega(\sqrt{Q^*_{A\rightarrow B}})$ on the communication plus entanglement cost of CDQS.
This seems to be a quantum analogue of the classical result giving a linear lower bound from communication complexity for a class of CDS protocols known as linear protocols \cite{gay2015communication}. 
This also has the implication that CDQS protocols that saturate the $\log Q^*_{A\rightarrow B}$ lower bound must use non-Clifford operations.\footnote{In future work, it would be interesting to explore this or similar observations as a lower bound technique for the number of non-Clifford operations needed in certain quantum operations.}

\subsection*{Tightness of lower bound from quantum communication complexity}

Here, we show that the lower bound derived in \cref{thm:oneway} in terms of the communication complexity is tight.
\cite{applebaum2017conditional} demonstrated this in the classical setting by finding a function with CDS cost upper bounded by $O(\log n)$ and with linear one-way communication complexity. 
We show in this section that the same function has linear quantum one-way communication complexity. 
It also immediately inherits the logarithmic upper bound on CDQS from the classical CDS protocol. 

Before stating the lower bound, we recall the concept of $\epsilon$-approximate degree of a function. 
\begin{definition}
    Let $f:\{0,1\}^n\to \{0,1\}$ be a Boolean function. 
    The $\epsilon$-approximate degree of $f$, denoted by $\mathrm{deg}_{\epsilon}(f)$, is defined as a minimum degree of a polynomial function $p:\{0,1\}^n \to \mathbb{R}$ satisfying 
    \begin{equation}
        \max_{x \in \{0,1\}^n} \left|f(x)-p(x)\right| \leq \epsilon \,. 
    \end{equation}
\end{definition}
From \cite[Theorem 1.1]{Sherstov2008pattern_journal}, we have the following lemma, which sets a lower bound on the quantum communication complexity of a class of functions. 
\begin{lemma}
~\label{lem:lower_bound}
Let $m$, $l$ be positive integers. Let $f': \{0,1\}^m \to \{0,1\}$ be a Boolean function. 
Define $f:\{0,1\}^{ml} \times \{0,1,\ldots,l-1\}^{m} \to \{0,1\}$ in the following way. 
Given an input $(x,y) \in \{0,1\}^{ml} \times \{0,1,\ldots,l-1\}^{m}$, divide $x$ into $m$ length-$l$ blocks. For each $0 \leq i \leq m-1$, choose the bit $x_{i,y_i}$ where $y_i$ is the $i$ th letter of $y$. 
Let $x_y$ be the resulting bitstring, and we define
\begin{equation}
    f(x,y) \coloneqq f'(x_y) \,. 
\end{equation}
Then, for any $\epsilon \in [0,1)$ and any $\delta \in [0,\epsilon/2)$, we have
\begin{align}
    Q^*_{\delta}(f) \geq \frac{1}{4}\mathrm{deg}_{\epsilon}(f') \log l - \frac{1}{2}\log \left(\frac{3}{\epsilon - 2\delta}\right) \,. 
\end{align}
\end{lemma} 
We apply \cref{lem:lower_bound} to the collision problem defined below. 
\begin{definition}
    The Collision Problem $\mathsf{Col}_n:\{0,1\}^{n\log n} \to \{0,1\}$ is a promise problem defined as follows. 
    Given an input $x \in \{0,1\}^{n\log n}$, divide $x$ into $n$ blocks of $\log n$ bits each. 
    For each $x$, define a function $f_x : \{0,1\}^{\log n} \to \{0,1\}^{\log n}$, where $f_x(i)$ is the $i$-th block of $x$. 
    Then, we define 
    \begin{equation}
        \mathsf{Col}_n(x) = \begin{cases}
            1 & f_x\,\,\mathrm{is\,\,a\,\,permutation}, \\ 
            0 &f_x \,\,\mathrm{is\,\,two\mhyphen to \mhyphen one}. 
        \end{cases}
    \end{equation}
     If $f_x$ is neither a permutation nor two-to-one, $\mathsf{Col}_n$ is undefined.
\end{definition}

We further define a variant of the collision problem. 
\begin{definition}
The problem $\mathsf{PCol}_n:\{0,1\}^{4n\log n} \times \{0,1,2,3\}^{n\log n} \to \{0,1\}$ is defined as follows. 
Given an input $(x,y) \in \{0,1\}^{4n\log n} \times \{0,1,2,3\}^{n\log n}$, divide $x$ into $n\log n$ length-$4$ blocks. 
For each $0 \leq i \leq n\log n-1$, choose the bit $x_{i,y_i}$ where $y_i$ is the $i$ th letter of $y$ 
and let $x_y$ be the resulting bitstring. 
We define
\begin{equation}
    \mathsf{PCol}_n(x,y) \coloneqq \mathsf{Col}_n(x_y) \,. 
\end{equation}
\end{definition}

By taking $m = n\log n$, $l = 4$, and $f' = \mathsf{Col}_n$ in \cref{lem:lower_bound}, we have the following proposition. 
\begin{proposition}
    For any $\epsilon \in [0,1)$ and any $\delta \in [0,\epsilon/2)$, we have 
\begin{align}
    Q^*_{\delta}(\mathsf{PCol}_n) \geq \frac{1}{2}\mathrm{deg}_{\epsilon}(\mathsf{Col}_n) - \frac{1}{2}\log \left(\frac{3}{\epsilon - 2\delta}\right) \,. 
\end{align}
\end{proposition}

From \cite{ambainis2005polynomial,kutin2005quantumlower}, we get that $\mathsf{PCol}_n$ 
has
\begin{align}
    \mathrm{deg}_{\epsilon}(\mathsf{Col}_n) =\Omega(n^{1/3}) \,,
\end{align}
so we get a polynomial lower bound on $Q^*_{\delta}(\mathsf{PCol}_n)$.
From our lower bound from quantum communication complexity (\cref{thm:oneway}), 
we have 
\begin{equation}
    \CDQS(\mathsf{PCol}_n) = \Omega(\log n) \,. 
\end{equation}
On the other hand, \cite{applebaum2017conditional} also puts an upper bound of $O(\log n)$ on $\CDS(\mathsf{PCol}_n)$, which puts an upper bound of $O(\log n)$ on $\CDQS(\mathsf{PCol}_n)$.
This gives matching upper and lower bounds on $\mathsf{PCol}_n$, and shows that we cannot get a super-logarithmic lower bound on CDQS in terms of $Q^
*$.  

\subsection{Lower bounds on perfectly private CDQS from \texorpdfstring{$\PP^{\cc}$}{TEXT}}

In this section we lower bound perfectly private CDQS in terms of the $\PP^{\cc}$ communication complexity. 

\begin{definition}[$\PP^{\cc}$]\label{def:PPcc}
A randomized communication protocol $\Pi$ involves two parties, Alice, who has input $x\in X$ and Bob who has input $y\in Y$. We say that $\Pi$ is a $\PP^{\cc}$ protocol for $f$ if, for every input $(x,y)$ the protocol outputs $f(x,y)$ with probability larger than or equal to $1/2+\beta$ with $\beta>0$. The cost of $\Pi$ is the total number of bits of communication $c$ plus the log of the inverse bias, $c+\log(1/\beta)$. The $\PP^{\cc}$ complexity of $f$, denoted by $\PP^{\cc}(f)$, is the minimum cost of any such protocol for $f$. 
\end{definition}

\begin{definition}[$\QPP^{\cc}$]\label{def:QPPcc}
A $\QPP^{\cc}$ protocol proceeds as a $\PP^{\cc}$ protocol, but now allows quantum messages. 
The $\QPP^{\cc}$ cost of a protocol is defined as the number of qubits of message exchanged plus the log of the inverse bias. 
We define $\QPP^{\cc}(f)$ as the minimal cost over all such protocols for the function $f$. 
We also consider allowing pre-shared entanglement, in which case we label the protocol a $\QPP^{*,\cc}$ protocol and the cost by $\QPP^{*,\cc}(f)$. 
\end{definition}
It was proven in \cite{klauck2001lower} (section 8) that $\PP^{\cc}(f)=\Theta(\QPP^{\cc}(f))$. 

The classical analogue of the lower bound that we would like to establish, proven in \cite{applebaum2021placing}, is as follows. 

\begin{theorem}[Reproduced from \cite{applebaum2021placing}] For every predicate $f : \{0, 1\}^{n} \times \{0, 1\}^{n} \rightarrow \{0, 1 \}$, 
\begin{align}
    \pp\CDS(f) \geq \Omega(\PP^{\cc}(f)) - O(\log(n))  \,.
\end{align}
\end{theorem}
Because $\PP^{\cc}(f)=\Theta(\QPP^{\cc}(f))$, we can again hope for a lower bound in terms of $\PP^{\cc}$ in the quantum setting. 
We will look therefore for a bound similar to the above but with the CDQS cost replacing the CDS cost. 

To this end, we make use of the following lemma, also invoked and proven in \cite{applebaum2021placing}.  
\begin{lemma}[Reproduced from \cite{applebaum2021placing}]\label{lemma:biasalgorithm}
    There exists a randomized algorithm A that given oracle access to a distribution $D_0$ and a distribution $D_1$ outputs 1 with probability exactly $1/2+ \Vert D_0-D_1\Vert_2^2/8$. Moreover, the algorithm uses three random bits and makes two non-adaptive queries to the oracles. 
\end{lemma}
Note that based on the random bits, the oracle calls can be either both to $D_0$, both to $D_1$, or one call to each. 

\begin{theorem}\label{thm:PPlowerbound} The communication cost of $\CDQS$ and the $\PP^{\cc}$ complexity are related by
\begin{align}\label{eq:PPlowerbound}
 \pp\overline{\CDQS}(f)=\Omega(\PP^{\cc}(f)) \,.
\end{align}
\end{theorem}
\begin{proof}\, We begin with a reduction of a CDQS protocol (here assumed perfectly private) to the appropriate communication complexity scenario. 
In the CDQS protocol, we let Alice and Bob's outputs be called $M_0$ and $M_1$ respectively, and denote $M \equiv M_0 M_1$. 
To define the one-way protocol, we have Alice prepare the entangled resource state used in the CDQS protocol, then send this to Bob. 
Alice and Bob then apply the first round operations defined by the CDQS taking the secret $Q$ to be a single qubit maximally entangled with a reference system $\bar{Q}$. 
Bob sends his output $M_1$ to Alice. 
We repeat this four times so that Alice obtains $\rho_{\bar{Q}M}^{\otimes 4}$. 
Alice then takes the third and fourth copies, traces out $\bar{Q}$, and replaces it with the maximally mixed state. 
Since also we always have $\rho_{\bar{Q}}=\mathcal{I}/d_{\bar{Q}}$, this prepares $\rho_{\bar{Q}} \otimes \rho_M$, so Alice holds two copies of $\rho_{\bar{Q}M}$ and two copies of $\rho_{\bar{Q}} \otimes \rho_M$. 

Recall that $\rho_{\bar{Q}M}$ and $\rho_{\bar{Q}} \otimes \rho_M$ will be identical in $f(x, y) = 0$ instances due to perfect privacy, whereas they will be different in $f(x, y) = 1$ instances due to correctness (assuming $\epsilon < 2$  so that the correctness criterion is not trivial). 
This means Alice's task is to distinguish between $\rho_{\bar{Q}M}$ and $\rho_{\bar{Q}} \otimes \rho_M$.

Alice's strategy is as follows. 
She applies a Haar random unitary $U_{\bar{Q}M}$ to each state on $\bar{Q}M$ (with the unitary fixed for all four states), then performs a measurement in the computational basis. 
This produces two samples each from distributions that we call $D_0(U)$ and $D_1(U)$, where $D_0(U)$ are measurement outcomes using the state $\rho_{\bar{Q}M}$ and $D_1(U)$ are using the state $\rho_{\bar{Q}}\otimes \rho_M$. 
There are sufficiently many samples to run the algorithm of \cref{lemma:biasalgorithm} above, which produces a binary variable $z$ equal to 1 with probability conditioned on $U$
\begin{align}
    p_1(U) = \frac{1}{2}+ \frac{1}{8} \lVert D_0(U) - D_1 (U) \rVert^2_2 \,.
\end{align}
A short calculation demonstrates that
\begin{equation}
    \int dU \lVert D_{0}(U) - D_{1}(U) \rVert_{2}^{2} = \frac{1}{d_{\bar{Q}M}+1} \lVert \rho_{\bar{Q}M} - \rho_{\bar{Q}} \otimes \rho_{M} \rVert_{2}^{2} \,.
\end{equation}
Consequently, the total probability of this procedure yielding $z=1$ is exactly
\begin{equation} \label{eq:prob1}
    p_{1} = \int dU p_{1}(U) = \frac{1}{2} + \frac{1}{8(d_{\bar{Q} M} + 1)} \lVert \rho_{\bar{Q}M} - \rho_{\bar{Q}} \otimes \rho_{M} \rVert_{2}^{2} \,.
\end{equation}

If Alice's final output was the value of the variable $z$, then the fact that $\rho_{\bar{Q}M}$ and $\rho_{\bar{Q}} \otimes \rho_{M}$ are equal if and only if $f(x, y) = 0$, along with \cref{eq:prob1}, implies that Alice's output will be correctly biased on $1$ instances, but unbiased on $0$ instances. 
To correct this asymmetry Alice should at the beginning of the protocol output $0$ with some small probability $s$ and otherwise perform the procedure described above. 
Doing so, we obtain
\begin{itemize}
    \item For $f(x, y) = 0$: $p_{0} = \frac{1+s}{2}$. 
    \item For $f(x, y) = 1$: $p_{1} = (1 - s) \left( \frac{1}{2} + \frac{1}{8(d_{\bar{Q} M} + 1)} \lVert \rho_{\bar{Q}M} - \rho_{\bar{Q}} \otimes \rho_{M} \rVert_{2}^{2} \right)$.
\end{itemize}
These will have the correct bias provided we take 
\begin{equation}
    0 < s < \frac{\lVert \rho_{\bar{Q}M} - \rho_{\bar{Q}} \otimes \rho_{M} \rVert_{2}^{2}}{4 (d_{\bar{Q} M} + 1) + \lVert \rho_{\bar{Q}M} - \rho_{\bar{Q}} \otimes \rho_{M} \rVert_{2}^{2} } \, ;
\end{equation}
the quantity on the right is monotonically increasing with $\lVert \rho_{\bar{Q}M} - \rho_{\bar{Q}} \otimes \rho_{M} \rVert_{2}^{2}$, so we would like to determine a lower bound on this quantity for $f(x, y)=1$ instances. This is furnished by correctness, which gives (we define $\pi_{Q \bar{Q}} = \mathcal{I}_{Q\bar{Q}}/d_{\bar{Q}Q}$)
\begin{align}
    2 \epsilon & \geq \lVert (\mathbfcal{D}^{x,y}_{M\rightarrow Q}\circ \mathbfcal{N}^{x,y}_{Q\rightarrow M})(\Psi_{Q \bar{Q}}^{+}) - \Psi_{Q \bar{Q}}^{+} \rVert_{1} + \lVert (\mathbfcal{D}^{x,y}_{M\rightarrow Q}\circ \mathbfcal{N}^{x,y}_{Q\rightarrow M})(\pi_{Q \bar{Q}}) - \pi_{Q \bar{Q}} \rVert_{1} \\
    & \geq \lVert \Psi_{Q \bar{Q}}^{+} - \pi_{Q \bar{Q}} \rVert_{1} - \lVert (\mathbfcal{D}^{x,y}_{M\rightarrow Q}\circ \mathbfcal{N}^{x,y}_{Q\rightarrow M}) (\Psi_{Q \bar{Q}}^{+} - \pi_{Q \bar{Q}} )\rVert_{1} \\
    & \geq \lVert \Psi_{Q \bar{Q}}^{+} - \pi_{Q \bar{Q}} \rVert_{1} - \lVert \mathbfcal{N}^{x,y}_{Q\rightarrow M} (\Psi_{Q \bar{Q}}^{+} - \pi_{Q \bar{Q}} )\rVert_{1} \\
    & = \frac{3}{2} - \lVert \rho_{\bar{Q}M} - \rho_{\bar{Q}} \otimes \rho_{M} \rVert_{1}
    \: ,
\end{align}
where $\Psi_{Q \bar{Q}}^{+}$ denotes the maximally entangled state on $Q \bar{Q}$. 
Here we have used the triangle inequality and the observation that the norm is non-increasing under the decoding channel. Recalling that an operator $A$ on a dimension $d$ Hilbert space satisfies (for example by H{\"o}lder's inequality)
\begin{equation}
    \lVert A \rVert_{1} \leq \sqrt{d} \lVert A \rVert_{2} \: ,
\end{equation}
we obtain
\begin{equation}
    \lVert \rho_{\bar{Q}M} - \rho_{\bar{Q}} \otimes \rho_{M} \rVert_{2}^{2} \geq \frac{1}{d_{\bar{Q}M}} \lVert \rho_{\bar{Q}M} - \rho_{\bar{Q}} \otimes \rho_{M} \rVert_{1}^{2} \geq \frac{1}{d_{\bar{Q} M}} \left( \frac{3}{2} - 2 \epsilon \right)^{2} \: .
\end{equation}
We therefore conclude that we may take $s$ satisfying
\begin{equation}
    0 < s < s_{0} \: , \qquad s_{0} \equiv \frac{\left( \frac{3}{2} - 2 \epsilon \right)^{2}}{4 d_{\bar{Q}M} (d_{\bar{Q} M} + 1) + \left( \frac{3}{2} - 2 \epsilon \right)^{2} } \: .
\end{equation}
Concretely we choose $s=s_0/2$. 

It remains to compute the resulting cost of this $\QPP^{\cc}$ protocol. 
The communication cost is $\pp\overline{\CDQS}(f)$ qubits to establish the needed shared entanglement, plus $\pp\CDQS(f)$ qubits for Bob to communicate his output system to Alice. 
We then need to add the logarithm of the inverse bias, which here is
\begin{align}
    s=O(1/d_M^2) \,,
\end{align}
so this gives an additional $2n_M = 2 \pp\CDQS(f)$ cost. 
The resulting bound is then as written in \cref{eq:PPlowerbound}. 
Note that we used that $\PP^{\cc}(f)=\Theta(\QPP^{\cc})$. 
\end{proof}

The same construction as in the proof above, now assuming Alice and Bob begin with the entangled resource state of the CDQS protocol, leads to the bound
\begin{align}
    \pp\CDQS(f)=\Omega(\QPP^{*,\cc}(f)) \,.
\end{align}
It would be interesting to better understand the relationship between $\QPP^{*,\cc}$ and $\PP^{\cc}$. 

\subsection{Lower bounds from quantum interactive proofs}\label{sec:QIPccbounds}

We first review the classical definition of an interactive proof in the communication complexity scenario. 

\begin{definition}[Reproduced from \cite{applebaum2021placing}]\label{def:IPcc}
 An $\IP^{\cc}$ protocol for a Boolean function $f:X\times Y\rightarrow \{0,1\}$ with $k'$ rounds proceeds as follows. 
Each round begins with a two-party protocol between Alice and Bob after which both parties send a message to Merlin, who sends back a message that is visible to both Alice and Bob. 
At the end of all $k'$ rounds, Alice and Bob interact again and generate an output. We say that the protocol accepts if both outputs equal $1$. 
The protocol is said to compute $f$ with completeness error of $\epsilon$ and soundness error of $\delta$ if it satisfies the following properties:
\begin{itemize}
    \item \emph{Completeness:} For all inputs $(x,y)$ with $f(x,y)=1$, there exists a proof strategy for Merlin such that $(x,y)$ is accepted with probability at least $1-\epsilon$. 
    \item \emph{Soundness:} For all inputs $(x,y)$ with $f(x,y)=0$, for any proof strategy for Merlin, $(x,y)$ is accepted with probability at most $\delta$. 
\end{itemize}
The cost of an $\IP^{\cc}$ protocol is the maximum over all inputs of the total communication complexity of the protocol. 
We also refer to a $k'$ round protocol as a $k=2k'$ message protocol.
The $\IP^{\cc}[k]$ complexity of $f$, denoted $\IP^{\cc}[k](f)$ is the smallest cost of a $k$-message $\IP^{\cc}$ protocol computing $f$ with soundness and completeness error of $\epsilon=\delta=1/3$. 
\end{definition}

The quantum definition requires only slight modifications. 

\begin{definition}\label{def:QIPcc}
\textbf{$\QIP^{\cc}$:} As in $\IP^{\cc}$, but now allowing quantum messages. Also, we have Merlin send his response (which can be quantum) to Alice only. 
We let $\QIP^{\cc}[k](f)$ denote the minimal quantum communication cost of a $k$ 
message $\QIP^{\cc}$ protocol for $f$. 
\end{definition}

With this definition in hand, we show that a good CDQS protocol leads to a good one round (two message) quantum interactive proof protocol. 

\begin{lemma}\label{lemma:CDQStoQIPcc}
    Suppose there is a CDQS protocol for $f$ using $t$ qubits of communication, $\rho$ qubits of shared entanglement, which hides $\ell$ bit secrets, and which is $\epsilon$ correct and $\epsilon$ secure.
    Then there is a two message quantum interactive proof protocol for $f$ which uses $t+\ell+1$ qubits of communication, $\rho$ qubits of entanglement and has completeness error of $\epsilon$ and soundness error of $\epsilon + 2^{-\ell}$. 
\end{lemma}

\begin{proof}\,
    Our proof closely follows the classical case \cite{applebaum2021placing}.
    Alice and Bob carry out the following $\QIP^{\cc}[2]$ protocol. 
    They share the same entangled state and execute the same first round operations as in the given CDQS protocol. 
    Additionally, they share $\ell$ bits of randomness in a string labelled $z$. 
    System $Q$ is prepared in the state $\ket{z}_Q$. 
    Alice and Bob then send their output systems to the referee, who sends back a string $z'$ to Alice. 
    Alice checks if $z=z'$, accepts if so, and sends Bob a single bit indicating that he should accept as well.

    First, consider why this is $\epsilon$ correct. 
    When $f(x,y)=1$, correctness implies that
    \begin{align}
        \Vert\mathbfcal{D}^{x,y}_{M\rightarrow Q}\circ \mathbfcal{N}^{x,y}_{Q\rightarrow M} - \mathbfcal{I}_{Q\rightarrow Q}\Vert_\diamond \leq \epsilon \,.
    \end{align}
    Inserting the input state $\ket{z}$, we get that the referee produces an output $\sigma_Q$ with $\Vert\sigma_Q-\ketbra{z}{z}_Q \Vert_1\leq \epsilon$, which by the Fuchs--van de Graaf inequality implies
    \begin{align}
        F(\sigma_Q, \ketbra{z}{z}_Q) = \bra{z}\sigma_Q\ket{z} \geq 1-\epsilon \,,
    \end{align}
    so the referee can correctly determine $z$ with probability at least $1-\epsilon$.
    When the referee returns $z$ Alice and Bob will accept, so we have $\epsilon$ correctness. 

    Next consider why this is $\epsilon+2^{-\ell}$ sound. 
    When $f(x,y)=0$, the security definition for CDQS ensures that Merlin's output is $\epsilon$ close to a state which is independent of $z$, 
    \begin{align}
        \Vert\sigma^0_Q - \sigma_Q(z)\Vert_1 \leq \epsilon \,.
    \end{align}
    Alice and Bob accept only if Merlin returns $z$, so they accept with probability
    \begin{align}
        p=\frac{1}{2^\ell} \sum_z \bra{z}\sigma_Q(z)\ket{z} \leq \frac{1}{2^\ell} \sum_z \bra{z}\sigma_Q^0\ket{z} + \epsilon = 2^{-\ell} + \epsilon \,,
    \end{align}
    as needed. 

    Notice that the communication cost is $t+\ell+1$ because Alice and Bob send the same messages they would have in the CDQS, which requires $t$ qubits,  Merlin sends back $\ell$ bits, and then Alice communicates to Bob whether or not to accept, costing an additional bit. 
    The entanglement is unchanged from the CDQS protocol. 
\end{proof}

Next, we want to lower bound the CDQS cost in terms of the $\QIP^{\cc}[2]$ cost. 
Recalling the definition of $\QIP^{\cc}[2]$, notice that we need to ensure we have correctness and soundness errors of at most $1/3$. 
If $\epsilon+2^{-\ell}>1/3$ or $\epsilon>1/3$ in our CDQS, the above lemma does not immediately lead to a $\QIP^{\cc}[2]$ protocol. 
We can resolve this however by first applying the amplification result of \cref{thm:amplification} to the CDQS, then applying \cref{lemma:CDQStoQIPcc}. 

\begin{theorem}\label{thm:QIPlowerbound}
The $\CDQS$ cost of a function $f$ is lower bounded asymptotically by the $\QIP^{\cc}[2]$ cost,
\begin{align}
    \CDQS(f) =\Omega( \QIP^{\cc}[2](f)) \,.
\end{align}
\end{theorem}
\begin{proof}\,
Given a CDQS with constant correctness and privacy error $\epsilon$ below a certain threshold, i.e. the threshold appearing in the definition of $\text{CDQS}(f)$, \cref{thm:amplification} allows us to reduce these to errors of order $\epsilon'=O(2^{-\ell})$ and increase the secret length to length $\ell$, while inducing overheads in entanglement and communication by a factor of $\ell$. 
Then for $\ell$ large enough $\epsilon'+2^{-\ell}<1/3$, so that \cref{lemma:CDQStoQIPcc} gives a valid $\QIP^{\cc}[2]$ protocol. 
\end{proof}

\subsection*{Zero knowledge quantum interactive proofs}

Our $\QIP^{\cc}[2]$ protocol has the interesting property that it is zero-knowledge in a particular sense --- each verifier doesn't learn anything about the others input. 
In this section we formalize this property by defining zero-knowledge quantum interactive proofs in the communication complexity setting and give a proof that the CDQS cost of a function $f$ is lower bounded by the cost of a zero-knowledge quantum interactive proof. 

\begin{definition}[$\HVQSZK^{\cc}$] An honest verifier quantum statistical zero-knowledge interactive proof is defined as follows. 
We let $\Pi$ be a $\QIP^{\cc}$ protocol for a Boolean function $f:\{0,1\}^n\times \{0,1\}^n\rightarrow \{0,1\}$. 
For inputs $(x,y)\in f^{-1}(1)$, let $\rho^{k'}_{AB}$ be the density matrix of the state held by Alice and Bob at the end of the $k'$-th round of $\Pi$. 
We consider simulator protocols $\Pi_S$ involving two parties $S_A$ holding $x$ and $S_B$ holding $y$, and allowing quantum communication between $S_A$ and $S_B$.
We divide the $\Pi_S$ protocol into $k'$ rounds, each of which can involve an arbitrary number of messages between $S_A$ and $S_B$. 
We say $\Pi_S$ is a $\delta$ simulation of $\Pi$ if after the $k'$-th round the simulators $S_A$ and $S_B$ hold a density matrix $\sigma^{k'}$ with $\Vert\rho^{k'}_{AB}-\sigma^{k'}_{AB}\Vert_1 \leq \delta$. 
We say $(\Pi, \Pi_S)$ is a $\HVQSZK^{\cc}$ protocol if $\delta<1/p(n)$ for any function $p(n)$ which is at least polylogarithmic in $n$.\footnote{We can understand this requirement as a communication complexity analogue of a function being negligible in the complexity setting.} 

To define the cost of a $\HVQSZK^{\cc}$ protocol, we first specify the notation:
\begin{itemize}
    \item Let $c_M$ be the bits sent between Alice and Merlin plus the qubits sent between Bob and Merlin in the protocol $\Pi$.
    \item Let $c_V$ be the qubits sent between Alice and Bob in the protocol $\Pi$.
    \item Let $c_S$ be the qubits sent between $S_A$ and $S_B$ in the protocol $\Pi_S$.
\end{itemize} 
Then we define 
\begin{align}
    \HVQSZK^{\cc}(f) = \min_{(\Pi,\Pi_S)}( c_M+\max\{c_V,c_S\} ) \,.
\end{align}
We similarly define the cost of a $k'$ round, $k=2k'$ message, protocol and denote it by $\HVQSZK^{\cc}[k](f)$.
\end{definition}

\begin{lemma}\label{lemma:CDQStoHVSZKcc}
The communication complexity of a $\CDQS$ protocol for function $f$ is lower bounded by the one round $\HVSZK^{\cc}$ cost according to
    \begin{align}
        \CDQS(f) \geq \Omega\left( \frac{\HVSZK^{\cc}[2](f)}{\log n}\right) \,.
    \end{align}
\end{lemma}

\begin{proof}\,
    We begin with a CDQS protocol with correctness and privacy errors $\epsilon \leq 0.09$, and which uses $t$ qubits of communication. 
    Apply the amplification result of \cref{thm:amplification} with $\ell=\alpha \log(n)$. 
    Then the resulting protocol hides $\Theta(\log(n))$ bit secrets, has correctness and privacy errors of $\epsilon'=\Theta(1/n^{\alpha})$, and has a communication cost of order $\Theta(t \log(n))$.

    From this CDQS protocol, we need to develop a $\QIP^{\cc}$ protocol $\Pi$ as well as a simulator protocol $\Pi_S$. 
    For the protocol $\Pi$, we use exactly the same construction as in \cref{lemma:CDQStoQIPcc}. 
    In particular, we have Alice and Bob prepare a $\ell = \alpha \log n$ bit random secret $z$, apply the same operations as in the CDQS, then send the resulting message system to Merlin. 
    Merlin responds with his guess $z'$ of the secret, Alice accepts herself and uses one extra bit to signal to Bob to accept if $z=z'$. 
    Recall that this has correctness error $\epsilon'$ and soundness error $\epsilon'+2^{-\ell}=O(1/n^\alpha)$, so for $n$ large enough this defines a valid $\QIP^{\cc}$ protocol (which requires correctness and soundness errors of less than 1/3).
    
    To define the simulator $\Pi_S$, we have $S_A$ carry out Alice's actions in the $\QIP^{\cc}$ protocol but now we trace out the message Alice sends to Merlin, and prepare a copy of $z$ in place of Merlin's response.
    $S_B$ carries out Bob's actions in the $\QIP^{\cc}$ protocol but now traces out the message Bob would have sent to Merlin. 
    Next, we claim that $\Pi_S$ is a $\delta=O(1/n^{\alpha/2})$ simulator of $\Pi$, and hence $(\Pi, \Pi_S)$ is a valid $\HVQSZK^{\cc}$ protocol.
    To check this notice that since we have a one-round protocol we only need to check the density matrices on Alice and Bob's systems is close to that produced in the $\QIP^{\cc}$ protocol after the first round. 
    Immediately after Merlin's response, the density matrix in the protocol $\Pi$ is $\rho=\frac{1}{2^{\ell}}\sum_{z}\rho_{AB}(z) \otimes \ketbra{z'(z)}{z'(z)}$, while in $\Pi_S$ it is $\sigma=\frac{1}{2^{\ell}}\sum_{z} \rho_{AB}(z) \otimes \ketbra{z}{z}$. 
    The trace distance is then
    \begin{align}
        \Vert\frac{1}{2^{\ell}}\sum_{z}\rho_{AB}(z) \otimes \ketbra{z'(z)}{z'(z)} -  \frac{1}{2^{\ell}}\sum_{z} \rho_{AB}(z) \otimes \ketbra{z}{z}\Vert_1 &\leq 2 \sqrt{1-F(\rho,\sigma)} \nonumber \\
        &= 2\sqrt{1-Pr[z=z']} \nonumber \\
        &\leq O(1/n^{\alpha/2}) \,.
    \end{align}
    as needed. 

    The simulator and real protocols both use $O(t \ell)=O(t\log(n))$ bits of communication to Merlin and no communication to one another, which leads to the stated lower bound. 
\end{proof}

A comment is that in the classical setting, a lower bound on CDS from a measure of the complexity of $\HVSZK$ that counts the randomness plus communication cost is proven.
This proof makes use of randomness sparsification, which we have no analogue for in the quantum setting. 
Because of this, we are so far limited to the above result which bounds CDQS in terms of the communication cost of a $\HVQSZK^{\cc}$ protocol.

\section{Discussion}\label{sec:discussion}

In this work, we studied the conditional disclosure of secrets primitive in the quantum setting.
Following the classical literature, we have proven closure and amplification results, and a number of lower bounds from communication complexity. 
With a single exception, we obtain close analogues of every known classical result in the quantum setting. 

The exceptional case is a result from \cite{applebaum2021placing} giving a lower bound on robust CDS from ${\AM}^{\cc}$. 
There is a quantum analogue of the relevant classical class (${\QAM}^{\cc}$), but the results used to derive this bound classically don't have a quantum analogue. 
In particular, the proof in \cite{applebaum2021placing} uses (a communication complexity analogue of) the fact that $\IP[k]\subseteq {\AM}[2]$ for all $k$, while the quantum analogue of this is not known to be true.
Thus while it is possible there is a lower bound on robust CDQS from ${\QAM}^{\cc}$, we do not see any clear route towards a proof given this discrepancy from the classical case. 

Perhaps the key distinction between classical and quantum CDS is the apparent lack of any connection between entanglement and communication cost in quantum CDS, compared to the upper bound on randomness from communication that exists for classical CDS. 
We leave as an open question whether such an ``entanglement sparsification'' lemma exists for quantum CDS. 
There are some reasons to not expect one however. 
For instance, there is no analogous statement for quantum communication complexity, even while there is one for classical communication complexity. 
Further, no such connection between entanglement and communication is known for the broader setting of non-local computation. 

Another basic question we leave open is a possible separation between classical and quantum CDS: does entanglement and/or quantum communication ever offer an advantage for CDS over randomness and classical communication?\footnote{After the initial appearance of this article, the subject of quantum-classical separations was taken up in \cite{girish2025comparing}.}

Perhaps the central direction that remains to be better understood is to further characterize the entanglement and communication cost of CDQS for general functions. 
In particular, none of the existing lower bound techniques can do better than establishing linear lower bounds, while the best upper bounds are $2^{O(\sqrt{n\log n})}$. 
While good explicit lower bounds are likely largely out of reach (since they imply circuit lower bounds), implicit lower bounds stated in terms of properties of $f(x,y)$ seem to be a viable target. 
For instance, we see no obvious obstruction to lower bounding $\CDQS(f)$ in terms of a function of the circuit complexity of $f(x,y)$. 
Doing so requires moving away from reductions to communication complexity settings however, and we do not know of any relevant techniques. 

\vspace{0.2cm}
\noindent \textbf{Acknowledgements:} This work was initiated at a retreat held by the Perimeter Institute for Theoretical Physics. 
We thank Alex Meiberg and Michael Vasmer for helpful discussions. 
Research at the Perimeter Institute is supported by the Government of Canada through the Department of Innovation, Science and Industry Canada and by the Province of Ontario through the Ministry of Colleges and Universities.

\appendix

\section{Clifford CDQS}\label{sec:CliffordCDQS}

We can consider a simplified setting where the resource system shared in a CDQS protocol is a stabilizer state, and Alice and Bob's operations are all Clifford. 
We refer to such protocols as \emph{Clifford $\CDQS$ protocols}.
In this appendix we show that Clifford CDS protocols sometimes need more communication and entanglement than general CDQS protocols. 
Concretely we prove the following lower bound. 

\begin{theorem}
    When restricting to Clifford $\CDQS$ protocols, we have
    \begin{align}
        \overline{\CDQS}(f) = \Omega(\sqrt{Q^*_{A\rightarrow B}}) \,.
    \end{align}
\end{theorem}
\begin{proof}\, (Sketch)
The proof follows the proof of \cref{thm:oneway}.
We suppose there is a Clifford CDQS for the function $f$ with entanglement plus communication cost $\overline{\CDQS}(f)$. 
To carry out the one-way quantum communication protocol, Alice and Bob share $k$ copies of the same entangled $\ket{\Psi}_{LR}$ state as they share in the CDQS. 
Upon receiving their inputs, Alice and Bob carry out isometric extensions of the first round operations they would have performed in the CDQS, and perform them on each of the $k$ copies. 
Alice then sends her output systems to Bob. 
Bob now holds a pure state $\ket{\Psi}_{QMR}^{\otimes k}$ where $M$ is the message system of the CDQS and $R$ is a purifying system. 
Bob performs tomography to learn $\ket{\Psi}_{QMR}$.  
Taking $k=n_Q+n_M+n_R$, Bob can learn $\ket{\Psi}_{QMR}$ exactly, but may fail with exponentially small probability \cite{montanaro2017learning}. 
Given that Bob learns $\ket{\Psi}_{QMR}$, he can use his classical description of the state to compute if $\rho_{QM}$ is close to product or far from product. 
If the state is close to product, he outputs $0$ as his guess of $f(x,y)$, while if the state is far from product he outputs $1$ as his guess of $f(x,y)$. 
This succeeds with probability exponentially close to $1$. 

The cost of this communication complexity protocol is $k$ times the size of $\ket{\Psi}_{QMR}$. 
Taking $Q$ to consist of $O(1)$ qubits, recalling that we took $k=n_Q+n_M+n_R$, and that $n_R\leq n_M+n_L$, we obtain that $Q^*_{A\rightarrow B}(f)\leq k^2 = (2n_M+n_L + n_Q)^2$, which is the needed lower bound on $\overline{\CDQS}(f) \geq n_M + n_L$. 
\end{proof}

\bibliographystyle{unsrtnat}
\bibliography{biblio}

\end{document}